\documentclass[12pt, draftclsnofoot, onecolumn]{IEEEtran}
\usepackage{epsfig,graphicx,subfigure,psfrag,amsmath,cases}
\usepackage{latexsym,amssymb,amsmath,epsfig,subfigure,algorithm,mathtools}
\usepackage{algorithmic,datetime}
\usepackage{color}
\usepackage{url}
\usepackage{scrtime}
\usepackage{array}
\usepackage{footnote}
\usepackage{multicol}
\usepackage{blindtext}
\usepackage[official]{eurosym}
\newtheorem{Thm}{Theorem}

\newtheorem{Remark}{Remark}
\DeclareMathOperator{\Tr}{\mathrm{Tr}}
\DeclareMathOperator{\zero}{\mathbf{0}}
\DeclareMathOperator{\Rank}{\mathrm{Rank}}

\DeclareMathOperator{\diag}{\mathrm{diag}}

\DeclareMathOperator{\maxo}{maximize}
\DeclareMathOperator{\mino}{minimize}

\newcommand{\qed}{\hfill \ensuremath{\blacksquare}}

\newcommand{\abs}[1]{\lvert#1\rvert}

\settimeformat{ampmtime}

\date{\currenttime ,\,\today}
\author{ Derrick Wing Kwan Ng, Marco Breiling, Christian Rohde, Frank Burkhardt, and Robert Schober\thanks{Derrick Wing Kwan Ng and Robert Schober are  with the Institute for Digital Communications, Friedrich-Alexander-University Erlangen-N\"urnberg (FAU), Germany (email:\{kwan, schober\}@lnt.de). Marco Breiling, Christian Rohde, and Frank Burkhardt are with Fraunhofer Institute for Integrated Circuits (IIS), Germany (email:\{marco.breiling, christian.rohde, frank.burkhardt\}@iis.fraunhofer.de).  This paper has been  presented  in part at  the $81$-st IEEE Vehicular Technology Conference, Glasgow, Scotland,  May 2015 \cite{CN:VTC2015}.
\vspace*{-2mm}
}}
\title{Energy-Efficient 5G Outdoor-to-Indoor Communication: SUDAS Over Licensed and Unlicensed Spectrum\vspace*{-2mm}}
\begin{document}
\maketitle\vspace*{-18mm}
\begin{abstract}\vspace*{-2mm}
In this paper, we study the joint resource allocation algorithm design  for  downlink and uplink
multicarrier transmission assisted by a shared user equipment (UE)-side distributed antenna system (SUDAS). The proposed SUDAS simultaneously
 utilizes  licensed frequency bands and unlicensed frequency bands, (e.g. millimeter wave bands),  to enable a spatial multiplexing gain for single-antenna UEs to improve  energy efficiency and system throughput of $5$-th generation (5G) outdoor-to-indoor communication.  The  design of the UE selection, the time allocation to uplink and downlink,  and the transceiver processing matrix is formulated as a non-convex optimization problem for the maximization of  the end-to-end system energy efficiency  (bits/Joule).  The proposed problem formulation takes into account  minimum data rate requirements for delay sensitive UEs and the circuit power consumption of all transceivers. In order to design a tractable resource allocation algorithm, we first show that the optimal transmitter precoding and receiver post-processing matrices  jointly diagonalize the end-to-end communication channel for both downlink and uplink communication via SUDAS. Subsequently,  the matrix optimization problem is converted to an equivalent scalar optimization problem for multiple parallel channels, which is solved by an asymptotically globally optimal iterative algorithm. Besides, we propose  a suboptimal algorithm which finds a locally optimal solution of the non-convex optimization problem.
 Simulation results illustrate that the proposed resource allocation algorithms for SUDAS
 achieve a significant performance gain  in terms of system energy efficiency and spectral efficiency compared to conventional baseline systems  by offering multiple parallel data streams for single-antenna UEs.  In fact, the proposed SUDAS
is able to bridge the gap between the current technology
and the high data rate and energy efficiency requirements of 5G outdoor-to-indoor
communication systems.
\end{abstract}
\begin{keywords}\vspace*{-4mm} 5G  outdoor-to-indoor communication, OFDMA resource allocation, non-convex optimization.
\end{keywords}

\section{Introduction}
\label{sect1}
High data rate,  high energy efficiency, and ubiquity are basic requirements for $5$-th generation (5G)  wireless communication systems.    A relevant technique for improving the system throughput for given quality-of-service
(QoS) requirements is    multiple-input multiple-output (MIMO) \cite{CN:VTC2015}--\nocite{JR:Massive_MIMO_mag}\cite{JR:TWC_large_antennas}, as it provides extra degrees of freedom in the spatial domain which facilitates a trade-off between multiplexing gain and diversity gain. In particular, massive MIMO, which equips the transmitter with a very large number of antennas  to serve a comparatively small number of user equipments (UEs), has received considerable interest recently \cite{JR:Massive_MIMO_mag,JR:TWC_large_antennas}. The  high flexibility in resource allocation makes massive MIMO a strong candidate for 5G communication systems.  However, state-of-the-art UEs are typically equipped with a small number of receive antennas which limits the spatial multiplexing gain offered by MIMO to individual UEs.
On the other hand, the combination of millimeter wave (mmW) and small cells, e.g. femtocells, has been proposed as a core network architecture for 5G indoor communication  systems \cite{JR:MMW_femto_cell,JR:MMW_femto_cell2} since most mobile data traffic is consumed indoors  \cite{Ericsson_report}.  The huge free, unlicensed frequency spectrum in the mmW frequency bands appears to be suitable and attractive for providing high speed communication services over short distances in the order of  meters. However, the  backhauling of the data from the service providers to  the small cell base stations (BSs) is a fundamental system bottleneck. In general,  the last mile connection from a backbone network to the UEs at homes can only support  high data rates if optical fibers are deployed, which is known as  fiber-to-the-home  (FTTH). Yet, the cost in deploying FTTH for all indoor users is prohibitive. For instance,  the cost in equipping every building  with FTTH in Germany is estimated to be around $67$ billion Euros \cite{FTTH} which makes high speed small cells not an appealing universal  solution for 5G indoor communication systems in terms of implementation  cost. Another attractive system architecture for 5G is to combine massive MIMO with mmW communications \cite{JR:massive_MIMO_mmw,JR:massive_MIMO_mmw2} by using outdoor mmW BSs.  However, the high penetration loss of  building walls limits the  suitability  of mmW for outdoor-to-indoor communication scenarios. Thus, additional effective system architecture for outdoor-to-indoor communication is needed.

 Distributed antenna systems  (DAS) are an existing system architecture on the network side and a special form of MIMO. DAS are able  to cover the dead spots in  wireless networks, extend service coverage, improve spectral efficiency, and  mitigate interference \cite{Mag:DAS,Mag:DAS2}.  It is expected that DAS will play an important role in 5G communication systems \cite{JR:VMIMO_heath}. Specifically, DAS can  realize the potential performance gains of MIMO systems by sharing antennas across  different terminals  of a communication system to form a virtual MIMO system \cite{VMIMO_thesis}. Lately, there has been a growing interest in combining orthogonal frequency division multiple access (OFDMA) and DAS  to pave the way  for the transition of existing communication systems to 5G  \cite{JR:virtual_MIMO1}--\nocite{JR:virtual_MIMO2}\cite{CN:DAS_OFDMA}.  In \cite{JR:virtual_MIMO1}, the authors studied suboptimal resource allocation algorithms for multiuser MIMO-OFDMA systems. In \cite{JR:virtual_MIMO2}, a utility-based low complexity scheduling scheme was proposed for multiuser MIMO-OFDMA systems to strike a balance between system throughput  and computational complexity.   The optimal subcarrier allocation, power allocation, and bit loading for  OFDMA-DAS was investigated in \cite{CN:DAS_OFDMA}. However, similar to massive MIMO, DAS cannot significantly improve  the  data rate of individual UE when the UEs are single-antenna devices.  Besides,  the
results in \cite{JR:virtual_MIMO1}--\nocite{JR:virtual_MIMO2}\cite{CN:DAS_OFDMA}, which are valid for either downlink or uplink communication, may no longer be applicable when joint optimization of downlink and uplink resource usage  is considered. Furthermore,  the total system throughput in \cite{JR:virtual_MIMO1}--\cite{CN:DAS_OFDMA} is not only limited by the number of antennas equipped at the individual UEs, but  is also constrained by the system bandwidth which is a very scarce resource in   licensed frequency bands.  In fact, licensed spectrum is usually located at sub $6$ GHz  frequencies which are suitable for long distance communication. On the contrary, the unlicensed frequency spectrum around $60$ GHz offers a large bandwidth of $7$ GHz for wireless communications but is more suitable for short distance communication. The simultaneous utilization of both  licensed  and unlicensed  frequency bands for  high rate communication introduces a paradigm shift in system and resource allocation algorithm design due to
the related new challenges and opportunities. Yet, the potential system throughput gains of such hybrid systems have not been thoroughly investigated  in the literature. Thus, in this work, we study the resource allocation design for  hybrid communication systems simultaneously utilizing  licensed and unlicensed frequency bands   to improve the system performance.

An important requirement for 5G systems is energy efficiency. Over the past decades, the development of wireless communication networks worldwide has triggered an
exponential growth in the number of wireless communication devices
for real time video teleconferencing, online high definition video streaming, environmental monitoring,  and safety management. It is expected that by $2020$, the number of interconnected devices on the planet may  reach up to $50$ billion \cite{IoT}. The related tremendous increase in the number of wireless communication transmitters and receivers has not only led to a huge demand for licensed bandwidth but also for energy. In particular, the escalating energy consumption of electronic circuitries for communication and radio frequency (RF) transmission  increases the operation cost of service providers and raises serious  environmental concerns due to the produced green house gases. As a result,   energy
efficiency  has become  as important as spectral efficiency   for evaluation of the performance of the resource utilization in communication networks. A tremendous number of green
resource allocation algorithm designs have been
proposed in the literature for
maximization of the energy efficiency  of wireless
communication systems \cite{JR:TWC_large_antennas,CR:virtual_MIMO0}\nocite{JR:EE_DAS}--\cite{JR:EE_DAS_letter}. In \cite{JR:TWC_large_antennas}, joint power allocation and subcarrier allocation was considered for energy-efficient massive MIMO systems.
  In \cite{CR:virtual_MIMO0}, the energy efficiency
 of a three-node  multiuser MIMO system  was studied for the two-hop
compress-and-forward relaying protocol. The trade-off between energy efficiency and spectral efficiency
  in DAS for fair resource allocation in flat fading channels was studied in \cite{JR:EE_DAS}. Power allocation for energy-efficient DAS was investigated in \cite{JR:EE_DAS_letter} for frequency-selective channels.
 However, in   \cite{JR:TWC_large_antennas,CR:virtual_MIMO0}--\cite{JR:EE_DAS_letter}, it was assumed that the transmit antennas were deployed by  service providers and are connected to a central unit by high cost optical  fibers or cables for facilitating simultaneous transmission  which may not be feasible in practice. To avoid this problem, unlicensed and licensed frequency bands may be used simultaneously to create a \emph{wireless data pipeline}  for  DAS to provide high rate communication services.  Nevertheless, the  resource allocation algorithm design for such a system architecture has not been investigated in the literature, yet.

 In this paper, we propose a  shared UE-side distributed antenna system (SUDAS) to assist the outdoor-to-indoor communication in 5G wireless communication systems.  In particular, SUDAS  simultaneously utilizes  licensed  and unlicensed  frequency bands to facilitate a spatial multiplexing gain for single-antenna transceivers. We formulate the resource allocation algorithm design for  SUDAS assisted OFDMA downlink/uplink transmission systems as a non-convex optimization problem.  By exploiting the structure of the optimal  precoding and  post-processing matrices adopted at  the BS and the SUDAS, the considered
matrix optimization problem is transformed into an equivalent optimization problem with scalar optimization variables. Capitalizing on this transformation, we develop an iterative algorithm
which achieves the asymptotically globally optimal performance of the proposed SUDAS for high signal-to-noise ratios (SNRs) and large numbers of subcarriers. Also, the asymptotically optimal algorithm serves as a building block for the design of a suboptimal resource allocation algorithm which achieves a locally optimal solution for the considered problem for arbitrary  SNRs.

\vspace*{-6mm}

\section{SUDAS Assisted OFDMA Network Model}\label{sect:OFDMA_AF_network_model}
 \begin{figure}[t]\vspace*{-3mm}
 \centering
\includegraphics[width=3.6in]{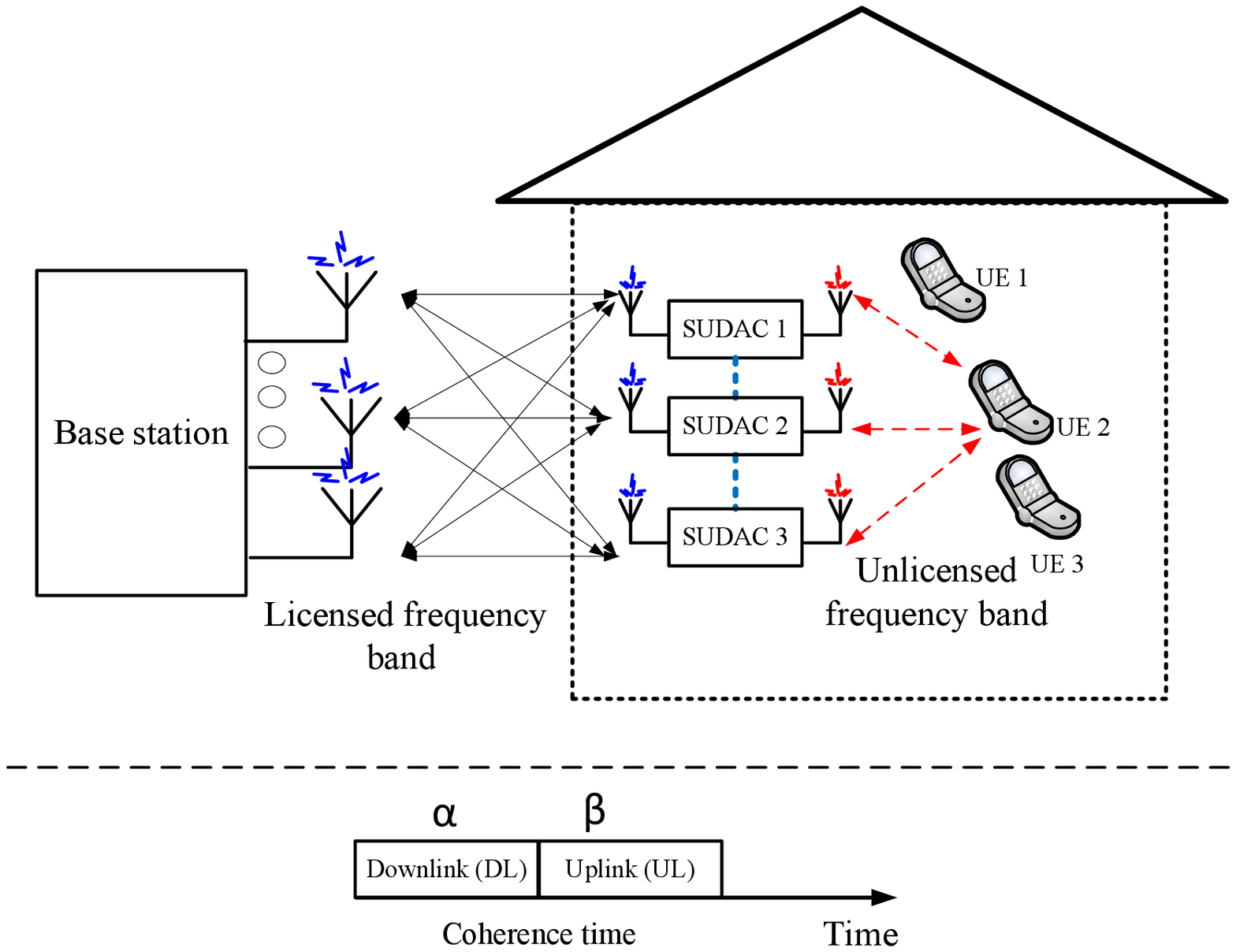}\vspace*{-5mm}
 \caption{The upper half of the figure illustrates the downlink and uplink  communication between  a base station (BS) and $K=3$  user equipments (UEs) assisted by $M=3$  SUDACs. The proposed system  utilizes a licensed frequency band and an unlicensed frequency band such as the mmW band (e.g. $\sim 60 $ GHz). The lower half of the figure depicts the time division duplex (TDD) approach adopted for downlink and uplink communication within a coherence time slot.  }\vspace*{-4mm}
 \label{fig:system_model}
\vspace*{-6mm}\end{figure}

\subsection{SUDAS System Model}
We consider a SUDAS assisted
OFDMA  downlink (DL) and uplink (UL)  transmission network which consists of one
$N$ antenna BS, a SUDAS, and $K$ single-antenna UEs, cf. Figure \ref{fig:system_model}.
  The  BS is half-duplex and equipped with  $N$ antennas for transmitting and receiving signals in a licensed frequency band.   The UEs are  single-antenna devices receiving and transmitting signals in the unlicensed frequency band. Also, we focus on a wideband multicarrier communication system with $n_{\mathrm{F}}$ orthogonal subcarriers. A SUDAS comprises  $M$ shared user equipment (UE)-side distributed antenna components (SUDACs). A SUDAC is a small and cheap device deployed inside a
   building\footnote{In practice, a SUDAC could be integrated into electrical devices such as electrical wall outlets, switches, and light outlets.} which simultaneously utilizes both a  licensed and an unlicensed frequency band for increasing the DL and UL end-to-end communication data rate. A basic SUDAC is  equipped with one antenna for use in a licensed band and one  antenna for use in an unlicensed band.  We note that the considered single-antenna model for SUDAC can be extended to the case of antenna arrays at the expense of a higher complexity and a more involved notation. Furthermore, a SUDAC is equipped with a mixer  to perform frequency up-conversion/down-conversion. For example, for DL communication, the SUDAC receives the signal from the BS in a licensed frequency band, e.g. at $800$ MHz, processes the received signal, and forwards the signal to the UEs in an unlicensed frequency band, e.g. the mmW bands. We note that since the BS-SUDAC link operates in a sub-6 GHz licensed frequency band, it is expected that the  associated path loss due to blockage by building walls is much smaller compared to the case where mmW bands were directly used for  outdoor-to-indoor communication.  Hence, the BS-to-SUDAS channel serves as wireless data pipeline for the SUDAS-to-UE communication channel.  Also, since  signal reception and transmission at each SUDAC are separated in frequency, cf. Figure \ref{fig:system_model2} and \cite{Online_sudas}, simultaneous signal reception and transmission  can be performed in the proposed SUDAS  which is not possible for traditional relaying systems\footnote{Since the BS-to-SUDAS and SUDAS-to-UE links operate in two different frequency bands, the proposed SUDAS should not be considered  a traditional relaying system \cite{JR:Jeff_7_ways}.} due to the limited spectrum availability in the licensed bands. The UL transmission via SUDAS can be performed in a similar manner as the DL transmission and the detailed operation will be discussed in the next section. In practice, a huge bandwidth is available in the unlicensed bands. For instance,  there is   nearly  $7$ GHz  unlicensed frequency spectrum available  for information transmission in the  $57 - 64$ GHz band (mmW bands).
 \begin{figure}[t]\vspace*{-3mm}
 \centering
\includegraphics[width=3.55in]{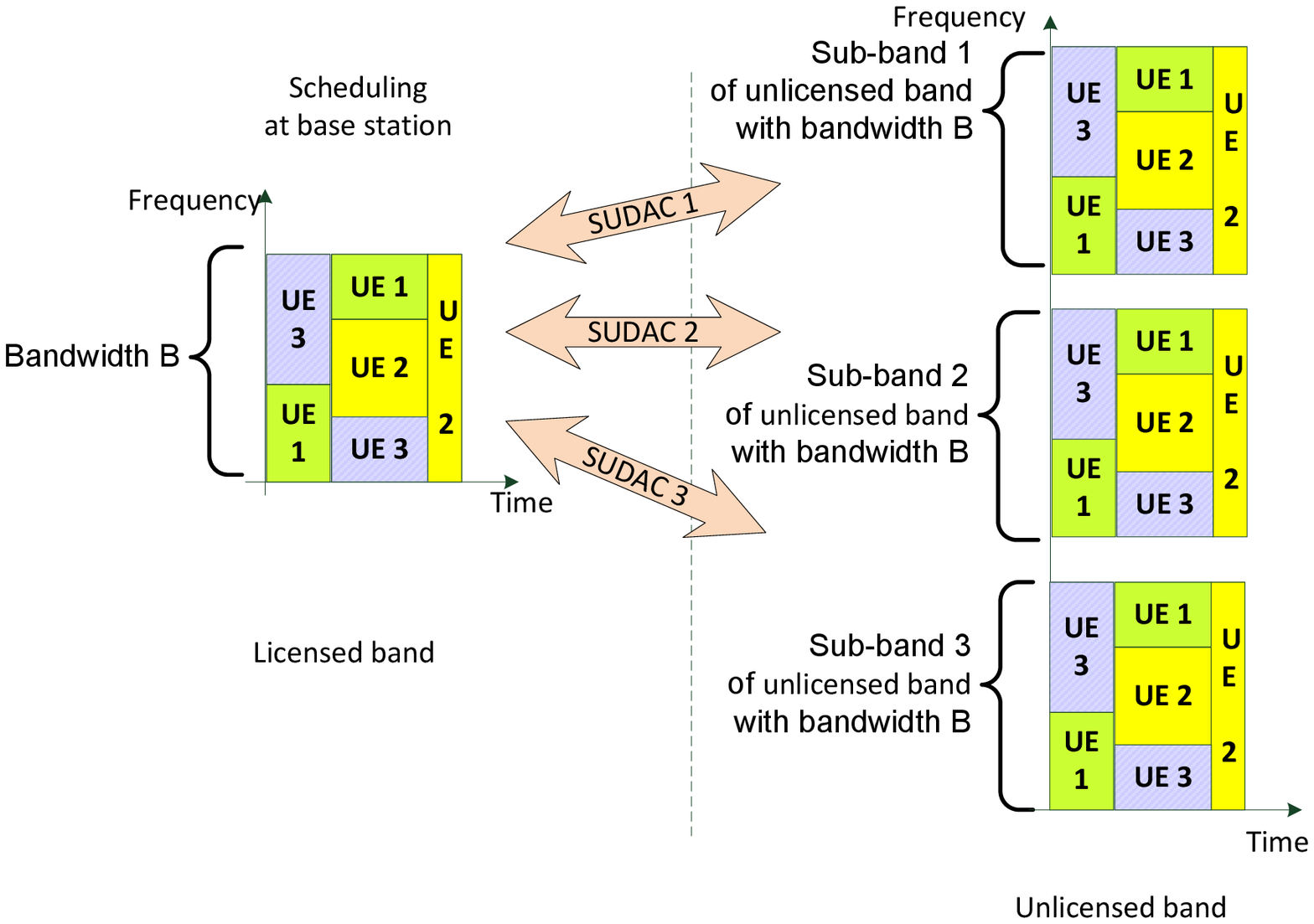}\vspace*{-6mm}
 \caption{Illustration of signal forwarding from(/to) a licensed band to(/from) different unlicensed frequency sub-bands in the SUDAS. }
 \label{fig:system_model2}
\vspace*{-8mm}\end{figure}
In this paper, we study the potential system performance gains for outdoor-to-indoor transmission achieved by the proposed SUDAS architecture. In particular, we focus on the case  where the SUDACs are installed in electrical wall outlets indoor and can cooperate with each other by sharing channel state information, power, and received signals, e.g. via  power line communication links. In other words, for the proposed resource allocation algorithm,  joint processing across the SUDACs is assumed to be possible such that the SUDACs can fully exploit the degrees of freedom offered by their antennas. The joint processing architecture  of the SUDAS in this paper reveals the maximum potential performance gain of the proposed SUDAS.

Furthermore, we adopt time division duplexing (TDD) to facilitate UL and DL communication for half-duplex UEs and BS. To simplify  the following presentation,  we assume a normalized unit length time slot whose duration is the coherence time of the channel,
 i.e., the communication channel is  time-invariant within a time slot. Each time slot is divided into two intervals of duration, $\alpha$ and $\beta$, which   are allocated for the DL and UL communication, respectively.

\vspace*{-4mm}
\subsection{SUDAS DL Channel Model}\label{sect:channel_model}
 \label{sect:channel_model}
In the DL transmission period  $\alpha$,   the BS performs \emph{spatial multiplexing} in the licensed band.  The data symbol vector $\mathbf{d}^{[i,k]}_{\mathrm{DL}}\in
\mathbb{C}^{N_{\mathrm{S}}\times 1}$ on subcarrier $i\in\{1,\,\ldots,\,n_F\}$
 for UE $k\in\{1,\,\ldots,\,K\}$ is precoded at the BS as\vspace*{-3mm}
\begin{eqnarray}
\mathbf{x}^{[i,k]}_{\mathrm{DL}}&=&\mathbf{P}^{[i,k]}_{\mathrm{DL}}\mathbf{d}^{[i,k]}_{\mathrm{DL}},\label{eqn:source_precoding}
\end{eqnarray}
where $\mathbf{P}^{[i,k]}_{\mathrm{DL}}\in\mathbb{C}^{N\times N_{\mathrm{S}}}$ is the
 precoding matrix adopted by the BS on subcarrier $i$ and $\mathbb{C}^{N\times N_{\mathrm{S}}}$ denotes the set of all $N\times N_{\mathrm{S}}$ matrices with complex entries.  The signals received on subcarrier $i$ at the $M$ SUDACs  for UE $k$ are given by
\begin{eqnarray}
\mathbf{y}^{[i,k]}_{\mathrm{S-DL}}&=&\mathbf{H}^{[i]}_{\mathrm{B}\rightarrow \mathrm{S}}\mathbf{x}^{[i,k]}_{\mathrm{DL}}+\mathbf{z}^{[i]},
\label{eqn:relay_channel_model:AF}
\end{eqnarray}
where $\mathbf{y}^{[i,k]}_{\mathrm{S-DL}}=[{y}^{[i,k]}_{\mathrm{S-DL}_1},\ldots, {y}^{[i,k]}_{\mathrm{S-DL}_M}]^T$, ${y}^{[i,k]}_{\mathrm{S-DL}_m}$ denotes the received signal at SUDAC $m\in\{1,\ldots,M\}$, and  $(\cdot)^T$ is the transpose operation. $\mathbf{H}^{[i]}_{\mathrm{B}\rightarrow \mathrm{S}} $ is the $ M\times  N$ MIMO channel
matrix between the BS and the $M$ SUDACs on subcarrier $i$ and captures
the joint effects of  path loss, shadowing, and multi-path fading.  $\mathbf{z}^{[i]}$ is the
additive white Gaussian noise (AWGN) vector  impairing the $M$ SUDACs in the licensed band on subcarrier $i$ and has a circularly symmetric complex Gaussian (CSCG) distribution ${\cal
CN}(\zero,\mathbf{\Sigma})$ on subcarrier $i$, where $\zero$ is the mean vector and
$\mathbf{\Sigma}$ is the $M\times M$
covariance matrix which is a diagonal matrix  with each main diagonal element given by $N_0$.

In the unlicensed band,  each SUDAC performs orthogonal \emph{frequency repetition}. In particular, the  $M$ SUDACs   multiply  the received signal
vector on subcarrier $i$, $\mathbf{y}^{[i,k]}_{\mathrm{S-DL}}$, by $\mathbf{F}^{[i,k]}_{\mathrm{DL}}\in\mathbb{C}^{M\times M}$ and forward
the processed signal vector to UE $k$ on subcarrier  $i$ in $M$  different independent frequency sub-bands in the unlicensed spectrum\footnote{For a signal bandwidth of $20$ MHz, there can be $350$ orthogonal sub-bands available within $7$ GHz of bandwidth in the $60$ GHz mmW band \cite{JR:massive_MIMO_mmw2}. For simplicity, we assume that each of the $M$ SUDACs uses one fixed sub-band for DL and UL communication.   }, cf. Figure \ref{fig:system_model2}.  In other words, different SUDACs forward that received signals in  different sub-bands and thereby avoid multiple access interference in the unlicensed spectrum.

The
signal  received  at UE $k$ on subcarrier $i$ from the SUDACs in the $M$ frequency bands, $\mathbf{y}^{[i,k]}_{\mathrm{S}\rightarrow\mathrm{UE}}\in\mathbb{C}^{M\times 1}$ ,  can be expressed as
\begin{eqnarray}
\mathbf{y}^{[i,k]}_{\mathrm{S}\rightarrow\mathrm{UE}}&=&\underbrace{\mathbf{H}^{[i,k]}_{\mathrm{S}\rightarrow\mathrm{UE}}\mathbf{F}^{[i,k]}_{\mathrm{DL}}
\mathbf{H}^{[i]}_{\mathrm{B}\rightarrow\mathrm{S}}\mathbf{P}^{[i,k]}_{\mathrm{DL}}\mathbf{d}^{[i,k]}_{\mathrm{DL}}}_{
\mbox{desired
signal}}+\underbrace{\mathbf{H}^{[i,k]}_{\mathrm{S}\rightarrow\mathrm{UE}}\mathbf{F}^{[i,k]}_{\mathrm{DL}}\mathbf{z}^{[i]}}_{\mbox{amplified
noise}}+\mathbf{n}^{[i,k]}.
\end{eqnarray}
The $m$-th element of vector $\mathbf{y}^{[i,k]}_{\mathrm{S}\rightarrow\mathrm{UE}}$ represents the received DL signal at UE $k$ in the $m$-th unlicensed frequency sub-band. Since the SUDACs forward the DL received signals in different orthogonal frequency bands, $\mathbf{H}^{[i,k]}_{\mathrm{S}\rightarrow\mathrm{UE}}$  is a diagonal matrix with the diagonal elements representing the channel gain between the SUDACs and UE $k$ on  subcarrier $i$ in  unlicensed sub-band $m$. $\mathbf{n}^{[i,k]}\in\mathbb{C}^{M\times 1}$ is the AWGN
vector at UE $k$ on subcarrier $i$  with distribution ${\cal
CN}(\zero,\mathbf{\Sigma}_k)$, where  $\mathbf{\Sigma}_k$
 is an ${M \times M}$ diagonal matrix and each main
diagonal element is equal to $N_{\mathrm{UE}_k}$.

We assume that $ M \ge N_{\mathrm{S}}$ and  UE $k$ employs a linear receiver for estimating the  DL data vector symbol received in the $M$ different sub-bands in the unlicensed band.
Hence, the estimated data vector symbol, $\mathbf{\hat{d}}^{[i,k]}_{\mathrm{DL}}\in\mathbb{C}^{N_\mathrm{S}\times 1}$, on subcarrier
$i$ at UE $k$ is given by
\begin{eqnarray}
\mathbf{\hat{d}}^{[i,k]}_{\mathrm{DL}}=(\mathbf{W}^{[i,k]}_{\mathrm{DL}})^H\mathbf{y}^{[i,k]}_{\mathrm{S}\rightarrow\mathrm{UE}},
\end{eqnarray}
where $\mathbf{W}^{[i,k]}_{\mathrm{DL}}\in \mathbb{C}^{M\times N_\mathrm{S}}$ is a
post-processing matrix used for subcarrier $i$ at UE $k$, and $(\cdot)^H$ denotes the Hermitian transpose. Without loss of generality, we assume that  ${\cal
E}\{\mathbf{d}^{[i,k]}_{\mathrm{DL}}(\mathbf{{d}}^{[i,k]}_{\mathrm{DL}})^H\}=\mathbf{I}_{N_\mathrm{S}}$ where $\mathbf{I}_{N_\mathrm{S}}$ is an $N_\mathrm{S}\times N_\mathrm{S}$ identity matrix and $\cal E\{\cdot\}$ denotes
statistical expectation. As a result,
the minimum mean square error (MMSE) matrix for data transmission on
subcarrier $i$ for UE $k$ via the proposed SUDAS and the optimal MMSE post-processing matrix are given by
\begin{eqnarray}
\mathbf{E}^{[i,k]}_{\mathrm{DL}}\hspace*{-2mm}&=&\hspace*{-2mm}{\cal E}\{(\mathbf{\hat{d}}^{[i,k]}_{\mathrm{DL}}-\mathbf{{d}}^{[i,k]}_{\mathrm{DL}})(\mathbf{\hat{d}}^{[i,k]}_{\mathrm{DL}}-
\mathbf{{d}}^{[i,k]}_{\mathrm{DL}})^H\}=\Big[\mathbf{I}_{N_\mathrm{S}}+(\mathbf{\Gamma}^{[i,k]}_{\mathrm{DL}})^H
(\mathbf{\Theta}^{[i,k]}_{\mathrm{DL}})^{-1}
\mathbf{\Gamma}^{[i,k]}_{\mathrm{DL}}\Big]^{-1},\\
\mbox{and }\mathbf{W}^{[i,k]}_{\mathrm{DL}}\hspace*{-2mm}&=&\hspace*{-2mm}(\mathbf{\Gamma}^{[i,k]}_{\mathrm{DL}}(\mathbf{\Gamma}^{[i,k]}_{\mathrm{DL}})^H
+\mathbf{\Theta}^{[i,k]}_{\mathrm{DL}})^{-1}\mathbf{\Gamma}^{[i,k]}_{\mathrm{DL}},
\label{eqn:AF-FD-self-interference}
\end{eqnarray}
respectively, where $(\cdot)^{-1}$ denotes the matrix inverse, $\mathbf{\Gamma}^{[i,k]}_{\mathrm{DL}}$ is the effective end-to-end channel
matrix from the BS to UE $k$ via the SUDAS on subcarrier $i$, and $\mathbf{\Theta}^{[i,k]}_{\mathrm{DL}}$
is the corresponding equivalent noise covariance matrix. These
matrices are given by
\begin{eqnarray}\label{eqn:AF-FD-equivalent_noise}
\mathbf{\Gamma}^{[i,k]}_{\mathrm{DL}}\hspace*{-2mm}&=&\hspace*{-2mm}
\mathbf{H}^{[i,k]}_{\mathrm{S}\rightarrow\mathrm{UE}}\mathbf{F}^{[i,k]}_{\mathrm{DL}}\mathbf{H}^{[i]}_{\mathrm{B}
\rightarrow\mathrm{S}}\mathbf{P}^{[i,k]}_{\mathrm{DL}}
\quad \mbox{and}\quad
\mathbf{\Theta}^{[i,k]}_{\mathrm{DL}}=
\Big(\mathbf{H}^{[i,k]}_{\mathrm{S}\rightarrow\mathrm{UE}}\mathbf{F}^{[i,k]}_{\mathrm{DL}}\Big)
\Big(\mathbf{H}^{[i,k]}_{\mathrm{S}\rightarrow\mathrm{UE}}\mathbf{F}^{[i,k]}_{\mathrm{DL}}\Big)^H+\mathbf{I}_{M}.
\end{eqnarray}
\begin{Remark}
The SUDAS concept is fundamentally different from traditional relaying systems which aim at extending service coverage \cite{JR:kwan_AF_relay,JR:MIMO_HD_relay1}. For DL communication, the SUDAS converts the spatial multiplexing performed at the BS in the licensed band into frequency multiplexing in the unlicensed band to allow  single-antenna UEs to decode multiple spatial data streams.
\end{Remark}

\vspace*{-4mm}
\subsection{SUDAS UL Channel Model}\label{sect:channel_model}
 \label{sect:channel_model}
In the UL transmission period $\beta$, UE $k$ performs \emph{frequency multiplexing} in the unlicensed band.  The data symbol vector $\mathbf{d}^{[i,k]}_{\mathrm{UL}}\in
\mathbb{C}^{N_{\mathrm{S}}\times 1}$ on subcarrier $i\in\{1,\,\ldots,\,n_{\mathrm{F}}\}$
 from UE $k$ is precoded as\vspace*{-2mm}
\begin{eqnarray}
\mathbf{x}^{[i,k]}_{\mathrm{UL}}=\mathbf{P}^{[i,k]}_{\mathrm{UL}}\mathbf{d}^{[i,k]}_{\mathrm{UL}},\label{eqn:source_precoding}
\end{eqnarray}
where $\mathbf{P}^{[i,k]}_{\mathrm{UL}}\in\mathbb{C}^{M\times N_{\mathrm{S}}}$ is the
 UL precoding matrix adopted by UE $k$ on subcarrier $i$ over the $M$ different frequency sub-bands in the unlicensed spectrum.  The signals received on subcarrier $i$ at the $M$ SUDACs  for UE $k$ are given by\vspace*{-2mm}
\begin{eqnarray}
\mathbf{y}^{[i,k]}_{\mathrm{S-UL}}&=&\mathbf{H}^{[i,k]}_{\mathrm{UE}\rightarrow \mathrm{S}}\mathbf{x}^{[i,k]}_{\mathrm{UL}}+\mathbf{v}^{[i]},
\label{eqn:relay_channel_model:AF}
\end{eqnarray}
where $\mathbf{y}^{[i,k]}_{\mathrm{S-UL}}=[{y}^{[i,k]}_{\mathrm{S-UL}_1}\ldots {y}^{[i,k]}_{\mathrm{S-UL}_M}]^T$, ${y}^{[i,k]}_{\mathrm{S-UL}_m}$ denotes the received signal at SUDAC $m$ in unlicensed frequency sub-band $m\in\{1,\ldots,M\}$, and $\mathbf{v}^{[i]}$ is the AWGN impairing the $M$ SUDACs  on subcarrier $i$ in the  unlicensed frequency band. $\mathbf{v}^{[i]}$ has distribution ${\cal CN}(\zero,\mathbf{\Sigma}_{\mathrm{UL}})$, where  $\mathbf{\Sigma}_{\mathrm{UL}}$
 is an ${M \times M}$ diagonal matrix and each main
diagonal element is equal to $N_{\mathrm{UL}}$. $\mathbf{H}^{[i,k]}_{\mathrm{UE}\rightarrow\mathrm{S}}$  is a diagonal matrix with the main diagonal elements representing the channel gains between UE $k$ and the $M$ SUDACs on  subcarrier $i$ in unlicensed sub-band $m$. In fact, the UEs-to-SUDAS channels serve as a short distance wireless data pipeline  for the SUDAS-to-BS UL communication.

Each SUDAC forwards the signals received  in the unlicensed band in the licensed band to assist the UL communication. In particular, the  $M$ SUDACs   multiply  the received signal
vector on subcarrier $i$ by $\mathbf{F}^{[i,k]}_{\mathrm{UL}}\in\mathbb{C}^{M\times M}$ and forward
the processed signal vector to the BS on subcarrier  $i$ in the licensed spectrum, cf. Figure \ref{fig:system_model2}.  As a result, the
signal  received  at the BS from UE $k$ on subcarrier $i$ via the SUDAS, $\mathbf{y}^{[i,k]}_{\mathrm{S}\rightarrow\mathrm{B}}\in\mathbb{C}^{N\times 1}$,  can be expressed as
\begin{eqnarray}
\mathbf{y}^{[i,k]}_{\mathrm{S}\rightarrow\mathrm{B}}\hspace*{-3mm}&=&\underbrace{\mathbf{H}^{[i]}_{\mathrm{S}\rightarrow\mathrm{B}}\mathbf{F}^{[i,k]}_{\mathrm{UL}}
\mathbf{H}^{[i,k]}_{\mathrm{UE}\rightarrow\mathrm{S}}\mathbf{P}^{[i,k]}_{\mathrm{UL}}\mathbf{d}^{[i,k]}_{\mathrm{UL}}}_{
\mbox{desired
signal}}+\underbrace{\mathbf{H}^{[i]}_{\mathrm{S}\rightarrow\mathrm{B}}\mathbf{F}^{[i,k]}_{\mathrm{UL}}\mathbf{z}^{[i]}}_{\mbox{amplified
noise}}+\mathbf{n}^{[i,k]}_{\mathrm{B}}.
\end{eqnarray}
Matrix  $\mathbf{H}^{[i]}_{\mathrm{S}\rightarrow\mathrm{B}}$ is the UL channel between the $M$ SUDACs and the BS on subcarrier $i$, and $\mathbf{n}^{[i,k]}_{\mathrm{B}}$  is the AWGN vector in subcarrier $i$ at the BS    with distribution ${\cal CN}(\zero,\mathbf{\Sigma}_{\mathrm{B}})$, where  $\mathbf{\Sigma}_{\mathrm{B}}$
 is an ${M \times M}$ diagonal matrix and each main
diagonal element is equal to $N_{\mathrm{B}}$.  At the BS,  we assume that $N\ge N_{\mathrm{S}}$ and the BS employs a linear receiver for estimating the  data vector symbol received from the SUDAS  in the licensed band.
The estimated data vector symbol, $\mathbf{\hat{d}}^{[i,k]}_{\mathrm{UL}}\in\mathbb{C}^{N_\mathrm{S}\times 1}$, on subcarrier
$i$ at the BS from  UE $k$ is given by
\begin{eqnarray}
\mathbf{\hat{d}}^{[i,k]}_{\mathrm{UL}}=(\mathbf{W}^{[i,k]}_{\mathrm{UL}})^H\mathbf{y}^{[i,k]}_{\mathrm{S}\rightarrow\mathrm{B}},
\end{eqnarray}
where $\mathbf{W}^{[i,k]}_{\mathrm{UL}}\in \mathbb{C}^{M\times N_{\mathrm{S}}}$ is a
post-processing matrix used for subcarrier $i$ at UE $k$. Without loss of generality, we assume that  ${\cal
E}\{\mathbf{d}^{[i,k]}_{\mathrm{UL}}(\mathbf{{d}}^{[i,k]}_{\mathrm{UL}})^H\}=\mathbf{I}_{N_\mathrm{S}}$. As a result,
the MMSE matrix for data transmission on
subcarrier $i$ from UE $k$ to the BS via the SUDAS and the optimal MMSE post-processing matrix are given by
\begin{eqnarray}
\mathbf{E}^{[i,k]}_{\mathrm{UL}}\hspace*{-2mm}&=&\hspace*{-2mm}{\cal E}\{(\mathbf{\hat{d}}^{[i,k]}_{\mathrm{UL}}-\mathbf{{d}}^{[i,k]}_{\mathrm{UL}})(\mathbf{\hat{d}}^{[i,k]}_{\mathrm{UL}}-
\mathbf{{d}}^{[i,k]}_{\mathrm{UL}})^H\}=\Big[\mathbf{I}_{N_\mathrm{S}}+(\mathbf{\Gamma}^{[i,k]}_{\mathrm{UL}})^H
(\mathbf{\Theta}^{[i,k]}_{\mathrm{UL}})^{-1}
\mathbf{\Gamma}^{[i,k]}_{\mathrm{UL}}\Big]^{-1},\\
\mbox{and }\mathbf{W}^{[i,k]}_{\mathrm{UL}}\hspace*{-2mm}&=&\hspace*{-2mm}(\mathbf{\Gamma}^{[i,k]}_{\mathrm{UL}}(\mathbf{\Gamma}^{[i,k]}_{\mathrm{UL}})^H
+\mathbf{\Theta}^{[i,k]}_{\mathrm{UL}})^{-1}\mathbf{\Gamma}^{[i,k]}_{\mathrm{UL}},
\label{eqn:AF-FD-self-interference}
\end{eqnarray}
respectively, where $\mathbf{\Gamma}^{[i,k]}_{\mathrm{UL}}$ is the effective end-to-end channel
matrix from UE $k$ to the BS via the SUDAS on subcarrier $i$, and $\mathbf{\Theta}^{[i,k]}_{\mathrm{UL}}$
is the corresponding equivalent noise covariance matrix. These
matrices are given by
\begin{eqnarray}
\mathbf{\Gamma}^{[i,k]}_{\mathrm{UL}}\hspace*{-2mm}&=&\hspace*{-2mm}
\mathbf{H}^{[i]}_{\mathrm{S}\rightarrow\mathrm{B}}\mathbf{F}^{[i,k]}_{\mathrm{UL}}\mathbf{H}^{[i,k]}_{\mathrm{UE}
\rightarrow\mathrm{S}}\mathbf{P}^{[i,k]}_{\mathrm{UL}}
\quad \mbox{and}\quad\mathbf{\Theta}^{[i,k]}_{\mathrm{UL}}=
\Big(\mathbf{H}^{[i,k]}_{\mathrm{UE}\rightarrow\mathrm{S}}\mathbf{F}^{[i,k]}_{\mathrm{UL}}\Big)
\Big(\mathbf{H}^{[i,k]}_{\mathrm{UE}\rightarrow\mathrm{S}}\mathbf{F}^{[i,k]}_{\mathrm{UL}}\Big)^H+\mathbf{I}_{M}.
\end{eqnarray}
\begin{Remark}
Since TDD is adopted and DL and UL transmission occur consecutively within the same coherence time,  for resource allocation algorithm design, it is reasonable to assume that channel reciprocity holds, i.e., $\mathbf{H}^{[i,k]}_{\mathrm{UE}\rightarrow\mathrm{S}}=(\mathbf{H}^{[i,k]}_{\mathrm{S}\rightarrow\mathrm{UE}})^H$  and $\mathbf{H}^{[i]}_{\mathrm{S}\rightarrow\mathrm{B}}=(\mathbf{H}^{[i]}_{\mathrm{B}\rightarrow\mathrm{S}})^H$.
\end{Remark}\vspace*{-6mm}
\section{Problem Formulation }\label{sect:cross-layer system}
In this section, we first introduce the adopted system performance measure. Then, the design of resource allocation and scheduling is formulated as an optimization problem.\vspace*{-4mm}
\subsection{System Throughput}
\label{subsect:Instaneous_Mutual_information}
The end-to-end DL and UL achievable data rate  on subcarrier $i$ between the BS and UE $k$ via the SUDAS
 are given by  \cite{JR:Yue_Rong_diagonalization}
  \begin{eqnarray}\label{eqn:cap_log_det}
R^{[i,k]}_{\mathrm{DL}}=-\log_2\Big(\det[\mathbf{E}^{[i,k]}_{\mathrm{DL}}]\Big)\quad \mbox{and}\quad R^{[i,k]}_{\mathrm{UL}}=-\log_2\Big(\det[\mathbf{E}^{[i,k]}_{\mathrm{UL}}]\Big),
\end{eqnarray}
respectively, where $\det(\cdot)$ is the determinant operation. The DL and UL data rate (bits/s)
for UE $k$  can be expressed as
\begin{eqnarray}
\label{eqn:user_TP}\rho^{[k]}_{\mathrm{DL}}=
\sum_{i=1}^{n_{\mathrm{F}}}s^{[i,k]}_{\mathrm{DL}}R^{[i,k]}_{\mathrm{UL}} \quad \mbox{and}\quad \rho^{[k]}_{\mathrm{UL}}=
\sum_{i=1}^{n_{\mathrm{F}}}s^{[i,k]}_{\mathrm{UL}}R^{[i,k]}_{\mathrm{DL}},
\end{eqnarray}
respectively, where  $s^{[i,k]}_{\mathrm{DL}}\in\{0,\alpha\}$ and $s^{[i,k]}_{\mathrm{UL}}\in\{0,\beta\}$ are the discrete subcarrier allocation
indicators, respectively. In particular,  a DL and an UL subcarrier can only be utilized for $\alpha$ and $\beta$ portions of the coherence time, respectively, or not be used at all.

 The   system throughput  is given by
 \begin{eqnarray}
\label{eqn:avg-sys-TP}{\cal U}({\cal P},{\cal
S})=\sum_{k=1}^K\rho^{[k]}_{\mathrm{DL}} +\sum_{k=1}^K \rho^{[k]}_{\mathrm{UL}}  \quad [\mbox{bits/s}] ,
\end{eqnarray}
where ${\cal
P}=\{\mathbf{P}^{[i,k]}_{\mathrm{DL}},\mathbf{F}^{[i,k]}_{\mathrm{DL}},\mathbf{P}^{[i,k]}_{\mathrm{UL}},
\mathbf{F}^{[i,k]}_{\mathrm{UL}}\}$ and
${\cal{S}}=\{s^{[i,k]}_{\mathrm{DL}},s^{[i,k]}_{\mathrm{UL}},\alpha,\beta\}$ are the precoding and subcarrier allocation
policies, respectively.

 On the other hand, the power consumption of the considered SUDAS assisted communication system consists of seven terms which can be divided into three groups and expressed as
\begin{subequations} \label{eqn:power_consumption}
\begin{eqnarray}
 \label{eqn:power_consumption1}
\hspace*{-8mm}{\cal U}_{\mathrm{TP}}({\cal P},{\cal
S})\hspace*{-2mm}&=&\hspace*{-2mm}\underbrace{P_{\mathrm{C}_{\mathrm{B}}}+NP_{\mathrm{Ant}_{\mathrm{B}}}+ M P_{\mathrm{C}_{\mathrm{SUDAC}}}+ K P_{\mathrm{C}_{\mathrm{UE}}}}_{\mbox{System circuit power consumption}}\\  \label{eqn:power_consumption2}
&+&\underbrace{\sum_{k=1}^{K}\sum_{i=1}^{n_{\mathrm{F}}}s^{[i,k]}_{\mathrm{DL}}\varepsilon_{\mathrm{B}}
\Tr\Big(\mathbf{P}^{[i,k]}_{\mathrm{DL}}(\mathbf{P}^{[i,k]}_{\mathrm{DL}})^H
\Big) +\sum_{k=1}^{K}\sum_{i=1}^{n_{\mathrm{F}}}s^{[i,k]}_{\mathrm{DL}}\varepsilon_{\mathrm{S}}
\Tr\Big(\mathbf{G}^{[i,k]}_{\mathrm{DL}}\Big)}_{\mbox{Total DL transmit  power consumption}} \\ \label{eqn:power_consumption3}
&+&\underbrace{\sum_{k=1}^{K}\varepsilon_k \sum_{i=1}^{n_{\mathrm{F}}} s^{[i,k]}_{\mathrm{UL}}\Tr\Big(\mathbf{P}^{[i,k]}_{\mathrm{UL}}(\mathbf{P}^{[i,k]}_{\mathrm{UL}})^H\Big) +\sum_{k=1}^{K}\sum_{i=1}^{n_{\mathrm{F}}}s^{[i,k]}_{\mathrm{UL}}
\varepsilon_{\mathrm{S}}\Tr\Big(\mathbf{G}^{[i,k]}_{\mathrm{UL}}\Big)}_{\mbox{Total UL   transmit  power consumption}} \quad [\mbox{Joule/s}]
\end{eqnarray}
\end{subequations}
\begin{eqnarray}\label{eqn:G1}
\mbox{where }\quad\mathbf{G}^{[i,k]}_{\mathrm{DL}}&=&
\mathbf{F}^{[i,k]}_{\mathrm{DL}}\Big(\mathbf{H}^{[i]}_{\mathrm{B}\rightarrow\mathrm{S}}\mathbf{P}^{[i,k]}_{\mathrm{DL}}
(\mathbf{P}^{[i,k]}_{\mathrm{DL}})^H(\mathbf{H}^{[i]}_{\mathrm{B}\rightarrow\mathrm{S}})^H
+
\mathbf{I}_M\Big)(\mathbf{F}^{[i,k]}_{\mathrm{DL}})^H,\\
\mathbf{G}^{[i,k]}_{\mathrm{UL}}&=&
\mathbf{F}^{[i,k]}_{\mathrm{UL}}\Big(\mathbf{H}^{[i,k]}_{\mathrm{UE}\rightarrow\mathrm{S}}\mathbf{P}^{[i,k]}_{\mathrm{UL}}
(\mathbf{P}^{[i,k]}_{\mathrm{UL}})^H(\mathbf{H}^{[i,k]}_{\mathrm{UE}\rightarrow\mathrm{S}})^H
+
\mathbf{I}_M\Big)(\mathbf{F}^{[i,k]}_{\mathrm{UL}})^H,\label{eqn:G2}
\end{eqnarray}
and $\Tr(\cdot)$ is the trace operator. The three positive constant terms in (\ref{eqn:power_consumption1}), i.e., $P_{\mathrm{C}_{\mathrm{B}}},$ $P_{\mathrm{C}_{\mathrm{SUDAC}}}$, and $ P_{\mathrm{C}_{\mathrm{UE}}}$,  represent the power dissipation of the circuits \cite{JR:limited_backhaul} for the basic operation of the BS, the SUDAC, and the UE, respectively, and  $P_{\mathrm{Ant}_{\mathrm{B}}}$ denotes the circuit power consumption per BS antenna. Equations \eqref{eqn:power_consumption2} and \eqref{eqn:power_consumption3} denote the total DL transmit power consumption and the total UL power consumption, respectively. Specifically,  $\Tr(\mathbf{G}^{[i,k]}_{\mathrm{DL}})$ and $\Tr(\mathbf{G}^{[i,k]}_{\mathrm{UL}})$  are the  DL and UL transmit powers of the SUDAS needed for facilitating the DL and UL communication of UE $k$ in subcarrier $i$, respectively. Similarly, $\Tr\Big(\mathbf{P}^{[i,k]}_{\mathrm{DL}}(\mathbf{P}^{[i,k]}_{\mathrm{DL}})^H\Big)$ and $\Tr\Big(\mathbf{P}^{[i,k]}_{\mathrm{UL}}(\mathbf{P}^{[i,k]}_{\mathrm{UL}})^H\Big)$ represent the DL transmit power from the BS to the SUDAS for UE $k$   and the UL transmit power from UE $k$ to the SUDAS in subcarrier $i$, respectively. To capture the power inefficiency
of power amplifiers, we introduce linear multiplicative constants $\varepsilon_{\mathrm{B}},\varepsilon_{\mathrm{S}},$ and $\varepsilon_k$ for the power radiated by the BS, the SUDAS, and UE $k$ in (\ref{eqn:power_consumption}), respectively. For instance, if
$\varepsilon_{\mathrm{B}}=4$, then for $1$ Watt of power radiated
in the RF, the BS  consumes $4$ Watt of power  which leads to a power amplifier efficiency of $25\%$\footnote{In this paper, we assume that  Class A power amplifiers with linear characteristic  are implemented at the transceivers. The maximum power efficiency of Class A amplifiers is limited to $25\%$.}.

The \emph{energy efficiency} of the considered system is defined as the
total  number of bits exchanged between the BS and the $K$ UEs via the SUDAS per Joule consumed energy:
\begin{eqnarray}
 \label{eqn:avg-sys-eff} \hspace*{-8mm}{\cal U}_{\mathrm{eff}}({\cal P},{\cal S})&=&
 \frac{{\cal U}({\cal P},{\cal
S})}{{\cal U}_{\mathrm{TP}}({\cal P},{\cal S})}\quad [\mbox{bits/Joule}].\vspace*{-4mm}
\end{eqnarray}
\subsection{Problem Formulation}\label{sect:cross-Layer_formulation}
 The optimal precoding matrices, ${\cal P}^*=\{\mathbf{P}^{[i,k]*}_{\mathrm{DL}},\mathbf{F}^{[i,k]*}_{\mathrm{DL}},\mathbf{P}^{[i,k]*}_{\mathrm{UL}},\mathbf{F}^{[i,k]*}_{\mathrm{UL}} \}$, and the optimal subcarrier allocation policy, ${\cal
S}^*=\{s^{[i,k]*}_{\mathrm{DL}},s^{[i,k]*}_{\mathrm{UL}},\alpha^*,\beta^*\}$, can be obtained by solving the following optimization problem:\vspace*{-4mm}
\begin{eqnarray}\label{eqn:cross-layer-formulation}
&&\hspace*{-3mm}\underset{{{\cal P},{\cal S}}}\maxo\ \
{\cal U}_{\mathrm{eff}}({\cal P},{\cal S}) \notag\\
\mathrm{s.t.} &\mbox{C1:}& \sum_{k=1}^{K}\sum_{i=1}^{n_{\mathrm{F}}}s^{[i,k]}_{\mathrm{DL}}\Tr\Big(\mathbf{P}^{[i,k]}_{\mathrm{DL}}(\mathbf{P}^{[i,k]}_{\mathrm{DL}})^H\Big) \le P_\mathrm{T},\notag\\
&\mbox{C2:}& \sum_{k=1}^{K}\sum_{i=1}^{n_{\mathrm{F}}}s^{[i,k]}_{\mathrm{DL}}\Tr\Big(\mathbf{G}^{[i,k]}_{\mathrm{DL}}\Big) \le M P_{\max}, \notag\\
&\mbox{C3:}& \sum_{i=1}^{n_{\mathrm{F}}}s^{[i,k]}_{\mathrm{UL}}\Tr\Big(\mathbf{P}^{[i,k]}_{\mathrm{UL}}(\mathbf{P}^{[i,k]}_{\mathrm{UL}})^H\Big) \le P_{\max_k},\forall k\in\{1,\ldots,K\},\notag\\
&\mbox{C4:}& \sum_{k=1}^{K}\sum_{i=1}^{n_{\mathrm{F}}}s^{[i,k]}_{\mathrm{UL}}\Tr\Big(\mathbf{G}^{[i,k]}_{\mathrm{UL}}\Big) \le  P_{\max}^{\mathrm{UL}}, \notag\\
&\mbox{C5:}& \rho^{[k]}_{\mathrm{DL}}\ge R^{\mathrm{DL}}_{\min_k} , \forall k\in {{\cal D}_\mathrm{DL}},\quad\mbox{C6: } \rho^{[k]}_{\mathrm{UL}}\ge  R^{\mathrm{UL}}_{\min_k} , \forall k\in {{\cal D}_\mathrm{UL}},\notag\\
&\mbox{C7:}&  \sum_{i=1}^{n_{\mathrm{F}}}s^{[i,k]}_{\mathrm{DL}}\le \alpha, \,\forall i,\hspace*{16.5mm} \mbox{C8:}\sum_{i=1}^{n_{\mathrm{F}}}s^{[i,k]}_{\mathrm{UL}}\le \beta, \quad\forall i,\notag\\
&\mbox{C9:}&s^{[i,k]}_{\mathrm{DL}}\in\{0,\alpha\}, \forall i,k,\notag \hspace*{10.5mm}\mbox{C10: }s^{[i,k]}_{\mathrm{UL}}\in\{0,\beta\},  \forall i,k,\notag\\
 &\mbox{C11:}& \alpha+\beta\le 1, \hspace*{25.5mm}\mbox{C12: } \alpha,\beta\geq 0.
\end{eqnarray}
Constants $P_{\mathrm{T}}$ and $M P_{\max}$ in C1 and C2 are the  maximum transmit power allowances for the BS and the SUDAS ($M$ SUDACs) for DL transmission, respectively, where $P_{\max}$ is the average transmit power budget for a SUDAC.   Similarly, constraints C3 and C4  limit the  transmit power  for UE $k$ and the SUDAS ($M$ SUDACs) for UL transmission, respectively, where $P_{\max_k}$ and $P_{\max}^{\mathrm{UL}}$ are the maximum transmit power budgets of UE $k$ and the SUDAS, respectively. We note that in practice the maximum transmit power allowances for the SUDAS-to-UE, $P_{\max}$, and SUDAS-to-BS, $P_{\max}^{\mathrm{UL}}$, may be different due to different regulations in licensed and unlicensed bands.  Sets ${{\cal D}_\mathrm{DL}}$ and ${{\cal D}_\mathrm{UL}}$ in constraints C5 and C6 denote the set of delay sensitive UEs for DL and UL communication, respectively. In particular, the system has to guarantee  a minimum required DL data rate $R^{\mathrm{DL}}_{\min_k}$ and UL data rate $R^{\mathrm{UL}}_{\min_k}$,  if  UE $k$ requests delay sensitive services in the DL and UL, respectively.  Constraints C7 -- C10 are imposed to
guarantee that each subcarrier can serve at most one UE for  DL and UL communication for fractions of $\alpha$ and $\beta$  of the available time. Constraints C11 and C12 are the boundary conditions  for the durations of DL and UL transmission.


\vspace*{-3mm}
\section{Resource Allocation Algorithm Design}
\label{transf_opt_prob}
The considered optimization problem has a non-convex objective function in fractional form. Besides, the precoding matrices $\{\mathbf{P}^{[i,k]}_{\mathrm{DL}},\mathbf{P}^{[i,k]}_{\mathrm{UL}}\}$ and $\{\mathbf{F}^{[i,k]}_{\mathrm{DL}},
\mathbf{F}^{[i,k]}_{\mathrm{UL}}\}$ are coupled in \eqref{eqn:G1} and \eqref{eqn:G2}  leading to a non-convex feasible solution set in \eqref{eqn:cross-layer-formulation}. Also, constraints C9 and C10 are combinatorial constraints which results in a discontinuity in the solution set.   In general, there is no systematic approach for solving non-convex optimization problems optimally. In many cases, an exhaustive search
method may be needed to obtain the global optimal solution. Yet, applying such method to our problem will lead to prohibitively high computational complexity since the search space for the optimal solution grows exponentially with respect to $K$ and $n_{\mathrm{F}}$.  In order to make the problem tractable, we first transform the objective function in fractional form into  an equivalent objective function in subtractive form via fractional programming theory. Subsequently, majorization theory is exploited to obtain the  structure of the optimal precoding policy to further simplify the problem. Then, we employ constraint relaxation to handle the binary constraints C9 and C10 to obtain an asymptotically optimal resource allocation algorithm in high SNR regime and for large numbers of subcarriers.

\subsection{Transformation of the Optimization Problem}
  For notational simplicity, we define $\mathcal{F}$ as the set of
feasible solutions of the optimization problem in
(\ref{eqn:cross-layer-formulation}) spanned by constraints C1 -- C12.  Without loss of generality,  we assume that $\{{\cal P},{\cal S}\}\in\mathcal{F}$ and  the solution set $\mathcal{F}$ is non-empty and compact. Then, the
maximum energy efficiency of the SUDAS assisted communication, denoted as $\eta_{\mathrm{eff}}^*$,   is given by
\begin{eqnarray}
\eta_{\mathrm{eff}}^*= \frac{{\cal U}({\cal P^*},{\cal
S^*})}{{\cal U}_{\mathrm{TP}}({\cal P^*},{\cal S^*})}=\underset{ \{{\cal P},{\cal S}\}\in\mathcal{F}}{\maxo}\,\frac{{\cal U}({\cal P},{\cal
S})}{{\cal U}_{\mathrm{TP}}({\cal P},{\cal S})}.
\end{eqnarray}

Now, we introduce the following  theorem for handling  the  optimization problem in \eqref{eqn:cross-layer-formulation}.
\begin{Thm}\label{Thm:1}
By nonlinear fractional programming theory  \cite{JR:fractional},  the resource allocation policy achieves the maximum energy efficiency $\eta_{\mathrm{eff}}^*$ if and only if it satisfies
\begin{eqnarray}\notag
\underset{ \{{\cal P},{\cal S}\}\in\mathcal{F}}{\maxo} \,\,{\cal U}({\cal P},{\cal
S})-\eta_{\mathrm{eff}}^*\,{\cal U}_{\mathrm{TP}}({\cal P},{\cal S})= {\cal U}({\cal P^*},{\cal
S^*})-\eta_{\mathrm{eff}}^*\,{\cal U}_{\mathrm{TP}}({\cal P^*},{\cal S^*})=0,
 \label{eqn:thm1}
\end{eqnarray}
\end{Thm}
for ${\cal U}({\cal P},{\cal
S})\ge0$ and ${\cal U}_{\mathrm{TP}}({\cal P},{\cal S})>0$.

 \emph{\,Proof:} Please refer to  \cite{JR:fractional} for a proof of Theorem 1. \qed

Theorem \ref{Thm:1} states the necessary and sufficient condition for a  resource allocation policy to be globally optimal. Hence, for an optimization problem
with an objective function in fractional form, there exists an
equivalent optimization problem with an
objective function in subtractive form, e.g. ${\cal U}({\cal P^*},{\cal
S^*})-\eta_{\mathrm{eff}}^*\,{\cal U}_{\mathrm{TP}}({\cal P^*},{\cal S^*})$ in this paper,  such that  the same optimal resource allocation policy solves both problems. Therefore, without loss of generality,  we can focus on the objective function in equivalent subtractive form to design a resource allocation policy which satisfies Theorem \ref{Thm:1} in the sequel.

\vspace*{-3mm}
\subsection{Asymptotically Optimal Solution}\label{sect:optimal_solution}
In this section, we propose an asymptotically optimal iterative algorithm  based on the
Dinkelbach method \cite{JR:fractional}  for solving
\eqref{eqn:cross-layer-formulation} with an equivalent objective function such that the obtained solution satisfies the conditions stated in Theorem  \ref{Thm:1}. The
proposed iterative  algorithm is summarized in Table \ref{table:algorithm1} (on the next page) and
the convergence to the optimal energy efficiency is guaranteed if the inner problem (\ref{eqn:inner_loop}) is solved in each iteration. Please refer to  \cite{JR:fractional} for a proof of the convergence of the iterative algorithm.

The iterative algorithm is implemented with a repeated loop. In each iteration in the
main loop, i.e., lines $3-10$,  we solve the following optimization problem  for a given
parameter $\eta_{\mathrm{eff}}$:
\begin{eqnarray}\label{eqn:inner_loop}\vspace*{-4mm}
&&\hspace*{-18mm}
\underset{ {{\cal P},{\cal S}}}{\maxo} \,\,{\cal U}({\cal P},{\cal
S})-\eta_{\mathrm{eff}}{\cal U}_{\mathrm{TP}}({\cal P},{\cal S})\nonumber\\
&&\hspace*{-6mm}\mathrm{s.t.} \,\,\mbox{C1 -- C12}.
\end{eqnarray}

\begin{table}[t]\vspace*{-5mm}\caption{Iterative Resource Allocation Algorithm.}\label{table:algorithm1}\vspace*{-12mm}
\begin{algorithm} [H]                 
\caption{Iterative Resource Allocation Algorithm }          
\label{alg1}                           
\begin{algorithmic} [1] \label{algorithm1}
\STATE Initialization: $L_{\max}=$ the maximum number of iterations  and  $\Delta\rightarrow 0$ is the
maximum tolerance
 \STATE Set $\eta_{\mathrm{eff}}=0$ and iteration index $t=0$

\REPEAT [Iteration Process: Main Loop]
\STATE For
a  given $\eta_{\mathrm{eff}}$, solve  (\ref{eqn:inner_loop}) and obtain an intermediate resource allocation policy $\{{\cal P'},{\cal S'}\}$
\IF {$\abs{{\cal U}({\cal P'},{\cal S'})-\eta_{\mathrm{eff}} {\cal U}_{\mathrm{TP}}({\cal P'},{\cal S'})<\Delta}$} \STATE  $\mbox{Convergence}=\,$\TRUE,\quad \textbf{return}
$\{{\cal P^*,\cal S^*}\}=\{{\cal P',\cal S'}\}$ and $\eta_{\mathrm{eff}}^*=\frac{{\cal U}({\cal
P'},{\cal S'})}{ {\cal U}_{\mathrm{TP}}({\cal P'},{\cal S'})}$
 \ELSE \STATE
Set $\eta_{\mathrm{eff}}=\frac{{\cal U}({\cal P'},{\cal S'})}{{\cal U}_{\mathrm{TP}}({\cal
P'},{\cal S'})}$ and $t=t+1$,  convergence $=$ \FALSE
 \ENDIF
 \UNTIL{Convergence $=$ \TRUE $\,$or $t=L_{\max}$}

\end{algorithmic}
\end{algorithm}
\vspace*{-1.6cm}
\end{table}

\subsubsection*{Solution of the Main Loop Problem (\ref{eqn:inner_loop})}
The transformed objective function is in subtractive form and is parameterized by variable $\eta_{\mathrm{eff}}$. Yet, the transformed problem is still a non-convex optimization problem. We handle the coupled precoding matrices by  studying the structure of the optimal precoding matrices for \eqref{eqn:inner_loop}.  In this context, we define the following matrices
to facilitate  the subsequent presentation. Using singular
value decomposition (SVD), the DL two-hop channel matrices $\mathbf{H}^{[i]}_{\mathrm{B}\rightarrow\mathrm{S}}$ and
$\mathbf{H}^{[i,k]}_{\mathrm{S}\rightarrow\mathrm{UE}}$ can be written as
 \begin{eqnarray}\label{eqn:SVD_HSR_HRD}
 \mathbf{H}^{[i]}_{\mathrm{B}\rightarrow\mathrm{S}}&=&\mathbf{U}^{[i]}_{\mathrm{B}\rightarrow\mathrm{S}}
 \mathbf{\Lambda}^{[i]}_{\mathrm{B}\rightarrow\mathrm{S}}(\mathbf{V}^{[i]}_{\mathrm{B}\rightarrow\mathrm{S}})^H \quad\mbox{and} \quad
 \mathbf{H}^{[i,k]}_{\mathrm{S}\rightarrow\mathrm{UE}}=
 \mathbf{U}^{[i,k]}_{\mathrm{S}\rightarrow\mathrm{UE}}\mathbf{\Lambda}^{[i,k]}_{\mathrm{S}\rightarrow\mathrm{UE}}
 (\mathbf{V}^{[i,k]}_{\mathrm{S}\rightarrow\mathrm{UE}})^H, \end{eqnarray}
 respectively, where $\mathbf{U}^{[i]}_{\mathrm{B}\rightarrow\mathrm{S}}\in{\mathbb{C}^{ M\times M}},\mathbf{V}^{[i]}_{\mathrm{B}\rightarrow\mathrm{S}}\in{\mathbb{C}^{N\times N}},
 \mathbf{U}^{[i,k]}_{\mathrm{S}\rightarrow\mathrm{UE}}\in{\mathbb{C}^{M\times M}},$ and  $\mathbf{V}^{[i,k]}_{\mathrm{S}\rightarrow\mathrm{UE}}\in{\mathbb{C}^{M\times M}}$
 are unitary matrices. $\mathbf{\Lambda}^{[i]}_{\mathrm{B}\rightarrow\mathrm{S}}$ and $\mathbf{\Lambda}^{[i,k]}_{\mathrm{S}\rightarrow\mathrm{UE}}$ and  are $ M \times  N$ and  $ M \times M$ matrices with main diagonal
 element vectors $\Big[\sqrt{\gamma_{\mathrm{B}\rightarrow\mathrm{S},1}^{[i]}}\, \sqrt{\gamma_{\mathrm{B}\rightarrow\mathrm{S},2}^{[i]}}\,\ldots\sqrt{\gamma_{\mathrm{B}\rightarrow\mathrm{S},R_1}^{[i]}}\Big]$
  and  $ \Big[\sqrt{\gamma_{\mathrm{S}\rightarrow\mathrm{UE},1}^{[i,k]}}\,
  \sqrt{\gamma_{\mathrm{S}\rightarrow\mathrm{UE},2}^{[i,k]}}\,$ $\ldots\,
  \sqrt{\gamma_{\mathrm{S}\rightarrow\mathrm{UE},R_2}^{[i,k]}}\Big]$, respectively, and all other elements equal to zero. Subscript indices $R_1=\Rank(\mathbf{H}^{[i]}_{\mathrm{B}\rightarrow\mathrm{S}})$ and $R_2=\Rank(\mathbf{H}^{[i,k]}_{\mathrm{S}\rightarrow\mathrm{UE}})$ denote the rank of matrices $\mathbf{H}^{[i]}_{\mathrm{B}\rightarrow\mathrm{S}}$ and  $\mathbf{H}^{[i,k]}_{\mathrm{S}\rightarrow\mathrm{UE}}$, respectively.    Variables $\gamma_{\mathrm{B}\rightarrow\mathrm{S},n}^{[i]}$ and $\gamma_{\mathrm{S}\rightarrow\mathrm{UE},n}^{[i,k]}$ represent the equivalent channel-to-noise ratio (CNR) on spatial
   channel $n$ in subcarrier $i$ of the BS-to-SUDAS channel and the SUDAS-to-UE $k$ channel, respectively. Similarly, we can exploit channel reciprocity and apply SVD to the UL two-hop channel matrices which yields
    \begin{align}\label{eqn:SVD_HSR_HRD_UL}
 \mathbf{H}^{[i]}_{\mathrm{S}\rightarrow\mathrm{B}}=\mathbf{V}^{[i]}_{\mathrm{B}\rightarrow\mathrm{S}}
 (\mathbf{\Lambda}^{[i]}_{\mathrm{B}\rightarrow\mathrm{S}})^H(\mathbf{U}^{[i]}_{\mathrm{B}\rightarrow\mathrm{S}})^H \quad\mbox{and}\quad \mathbf{H}^{[i,k]}_{\mathrm{UE}\rightarrow\mathrm{S}}=\mathbf{V}^{[i,k]}_{\mathrm{S}\rightarrow\mathrm{UE}} (\mathbf{\Lambda}^{[i,k]}_{\mathrm{S}\rightarrow\mathrm{UE}} )^H(\mathbf{U}^{[i,k]}_{\mathrm{S}\rightarrow\mathrm{UE}})^H. \end{align}
 We are now ready to introduce the following theorem.
\begin{Thm}\label{Thm:Diagonalization_optimal}
Assuming that
$\Rank(\mathbf{P}^{[i,k]}_{\mathrm{DL}})=\Rank(\mathbf{P}^{[i,k]}_{\mathrm{UL}})=\Rank(\mathbf{F}^{[i,k]}_{\mathrm{DL}})=\Rank(\mathbf{F}^{[i,k]}_{\mathrm{UL}})=N_{\mathrm{S}}\le
 \min\{\Rank(\mathbf{H}^{[i,k]}_{\mathrm{S}\rightarrow\mathrm{UE}}),\Rank(\mathbf{H}^{[i]}_{\mathrm{B}\rightarrow\mathrm{S}})\}$, the optimal linear
  precoding matrices used at the BS and the SUDACs for the maximization problem in \eqref{eqn:inner_loop} jointly diagonalize the DL and UL channels of the  BS-SUDAS-UE link on each subcarrier, despite the non-convexity of the objective function\footnote{We note that the diagonal structure is also optimal for frequency division duplex systems where $\mathbf{H}^{[i,k]}_{\mathrm{UE}\rightarrow\mathrm{S}}\neq(\mathbf{H}^{[i,k]}_{\mathrm{S}\rightarrow\mathrm{UE}})^H$ and $\mathbf{H}^{[i]}_{\mathrm{S}\rightarrow\mathrm{B}}\neq(\mathbf{H}^{[i]}_{\mathrm{B}\rightarrow\mathrm{S}})^H$. Only  the  optimal precoding matrices in \eqref{eqn:matrix_P} and \eqref{eqn:matrix_F} will  change accordingly to jointly diagonalize the end-to-end channel.}.  The optimal precoding matrices have the following structure:
\begin{eqnarray}\label{eqn:matrix_P}
\mathbf{P}^{[i,k]}_{\mathrm{DL}}&=&\mathbf{\widetilde V}^{[i]}_{\mathrm{B}\rightarrow\mathrm{S}}\mathbf{\Lambda}^{[i,k]}_{\mathrm{B}_\mathrm{DL}},
\quad \mbox{  }\quad\quad \,\,\,\,
\mathbf{F}^{[i,k]}_{\mathrm{DL}}=\mathbf{\widetilde V}^{[i,k]}_{\mathrm{S}\rightarrow\mathrm{UE}}\mathbf{\Lambda}^{[i,k]}_{\mathrm{F}_\mathrm{DL}}(\mathbf{\widetilde U}^{[i,k]}_{\mathrm{B}\rightarrow\mathrm{S}})^H,\\  \label{eqn:matrix_F}
\mathbf{P}^{[i,k]}_{\mathrm{UL}}&=&\mathbf{\widetilde U}^{[i]}_{\mathrm{S}\rightarrow\mathrm{UE}}\mathbf{\Lambda}^{[i,k]}_{\mathrm{UE}_\mathrm{UL}},
\,\,\, \mbox{and}\quad
\mathbf{F}^{[i,k]}_{\mathrm{UL}}=\mathbf{\widetilde U}^{[i,k]}_{\mathrm{B}\rightarrow\mathrm{S}}\mathbf{\Lambda}^{[i,k]}_{\mathrm{F}_\mathrm{UL}}(\mathbf{\widetilde V}^{[i,k]}_{\mathrm{S}\rightarrow\mathrm{UE}})^H,
      \end{eqnarray}
respectively, where $\mathbf{\widetilde V}^{[i]}_{\mathrm{B}\rightarrow\mathrm{S}}$, $\mathbf{\widetilde V}^{[i,k]}_{\mathrm{S}\rightarrow\mathrm{UE}}$, and $\mathbf{\widetilde U}^{[i,k]}_{\mathrm{B}\rightarrow\mathrm{S}}$ are the $N_{\mathrm{S}}$ rightmost columns of $\mathbf{V}^{[i]}_{\mathrm{B}\rightarrow\mathrm{S}}$, $\mathbf{ V}^{[i,k]}_{\mathrm{S}\rightarrow\mathrm{UE}}$, and $\mathbf{ U}^{[i,k]}_{\mathrm{B}\rightarrow\mathrm{S}}$, respectively.   Matrices $\mathbf{\Lambda}^{[i,k]}_{\mathrm{B}_\mathrm{DL}}\in
\mathbb{C}^{N_{\mathrm{S}}\times N_{\mathrm{S}}}$, $\mathbf{\Lambda}^{[i,k]}_{\mathrm{F}_\mathrm{DL}}\in
\mathbb{C}^{N_{\mathrm{S}} \times N_{\mathrm{S}}}$,  $\mathbf{\Lambda}^{[i,k]}_{\mathrm{B}_\mathrm{UL}}\in
\mathbb{C}^{N_{\mathrm{S}}\times N_{\mathrm{S}}}$, and  $\mathbf{\Lambda}^{[i,k]}_{\mathrm{F}_\mathrm{UL}}\in
\mathbb{C}^{N_{\mathrm{S}}\times N_{\mathrm{S}}}$ are diagonal matrices which can be expressed as
\begin{eqnarray}
\mathbf{\Lambda}^{[i,k]}_{\mathrm{B}_\mathrm{DL}}&=&\diag\Big(\sqrt{P_{\mathrm{B}\rightarrow\mathrm{S},1}^{[i,k]}}
\,\ldots\,\sqrt{P_{\mathrm{B}\rightarrow\mathrm{S},n}^{[i,k]}}\,\ldots\,\sqrt{P_{\mathrm{B}\rightarrow\mathrm{S},N_{\mathrm{S}}}^{[i,k]}}\Big),\\
\mathbf{\Lambda}^{[i,k]}_{\mathrm{F}_\mathrm{DL}}&=&\diag\Big(\sqrt{P_{\mathrm{S}\rightarrow\mathrm{UE},1}^{[i,k]}}
 \,\ldots\,\sqrt{P_{\mathrm{S}\rightarrow\mathrm{UE},n}^{[i,k]}}\,\ldots\,
\sqrt{P_{\mathrm{S}\rightarrow\mathrm{UE},N_{\mathrm{S}}}^{[i,k]}}\Big),\\
\mathbf{\Lambda}^{[i,k]}_{\mathrm{UE}_\mathrm{UL}}&=&\diag\Big(\sqrt{P_{\mathrm{UE}\rightarrow\mathrm{S},1}^{[i,k]}}
\,\ldots\,\sqrt{P_{\mathrm{UE}\rightarrow\mathrm{S},n}^{[i,k]}}\,\ldots\,\sqrt{P_{\mathrm{UE}\rightarrow\mathrm{S},N_{\mathrm{S}}}^{[i,k]}}\Big), \mbox{and}\\
\mathbf{\Lambda}^{[i,k]}_{\mathrm{F}_\mathrm{UL}}&=&\diag\Big(\sqrt{P_{\mathrm{S}\rightarrow\mathrm{B},1}^{[i,k]}}
\,\ldots\,\sqrt{P_{\mathrm{S}\rightarrow\mathrm{B},n}^{[i,k]}}\,\ldots\,\sqrt{P_{\mathrm{S}\rightarrow\mathrm{B},N_{\mathrm{S}}}^{[i,k]}}\Big),
\end{eqnarray}
 respectively,  where  $\diag(x_1, \cdots, x_K)$ denotes a diagonal matrix with the diagonal elements $\{x_1, \cdots, $ $ x_K\}$. Scalar optimization variables $P_{\mathrm{B}\rightarrow\mathrm{S},n}^{[i,k]}$, $P_{\mathrm{S}\rightarrow\mathrm{UE},n}^{[i,k]}$,  $P_{\mathrm{UE}\rightarrow\mathrm{S},n}^{[i,k]}$, and $P_{\mathrm{S}\rightarrow\mathrm{B},n}^{[i,k]}$ are, respectively, the equivalent
transmit powers of the BS-to-SUDAS link, the SUDAS-to-UE link, the UE-to-SUDAS link, and the SUDAS-to-BS link  for UE $k$ on spatial channel $n$ and subcarrier $i$.
\end{Thm}

\begin{proof}
Please refer to the Appendix.
\end{proof}

By adopting the optimal precoding matrices provided in Theorem \ref{Thm:Diagonalization_optimal}, the DL and UL end-to-end channel on subcarrier $i$ is converted into $N_{\mathrm{S}}$ parallel spatial channels. More importantly,  the structure of the optimal precoding matrices simplifies the resource allocation algorithm design significantly as the matrix optimization  variables can be replaced by equivalent scalar  optimization variables. As a result, the achievable rates in DL and UL on subcarrier $i$ from the BS to UE $k$ via the SUDAS in \eqref{eqn:cap_log_det}  can
be simplified as
\begin{eqnarray}\label{eqn:R1}
 \hspace*{-3mm}R^{[i,k]}_{\mathrm{DL}}\hspace*{-3mm}&=&\hspace*{-3mm}\sum_{n=1}^{N_{\mathrm{S}}}
 \log_2\Big(1+\mathrm{SINR}^{[i,k]}_{\mathrm{DL}_n}\Big),\,\,
\mathrm{SINR}^{[i,k]}_{\mathrm{DL}_n}=
\frac{\gamma_{\mathrm{B}\rightarrow\mathrm{S},n}^{[i]}P_{\mathrm{B}\rightarrow\mathrm{S},n}^{[i,k]}
P_{\mathrm{S}\rightarrow\mathrm{UE},n}^{[i,k]}\gamma_{\mathrm{S}\rightarrow\mathrm{UE},n}^{[i,k]}}{1+\gamma_{\mathrm{B}\rightarrow\mathrm{S},n}^{[i]}P_{\mathrm{B}\rightarrow\mathrm{S},n}^{[i,k]}
+ P_{\mathrm{S}\rightarrow\mathrm{UE},n}^{[i,k]}\gamma_{\mathrm{S}\rightarrow\mathrm{UE},n}^{[i,k]}},\\
\hspace*{-3mm} R^{[i,k]}_{\mathrm{UL}}\hspace*{-3mm}&=&\hspace*{-3mm}\sum_{n=1}^{N_{\mathrm{S}}}\log_2\Big(1+\mathrm{SINR}^{[i,k]}_{\mathrm{UL}_n}\Big),\,\,
\mathrm{SINR}^{[i,k]}_{\mathrm{UL}_n}=
\frac{\gamma_{\mathrm{S}\rightarrow\mathrm{B},n}^{[i]}P_{\mathrm{S}\rightarrow\mathrm{B},n}^{[i,k]}
P_{\mathrm{UE}\rightarrow\mathrm{S},n}^{[i,k]}\gamma_{\mathrm{UE}\rightarrow\mathrm{S},n}^{[i,k]}}{1+\gamma_{\mathrm{S}\rightarrow\mathrm{B},n}^{[i]}P_{\mathrm{S}\rightarrow\mathrm{B},n}^{[i,k]}
+ P_{\mathrm{UE}\rightarrow\mathrm{S},n}^{[i,k]}\gamma_{\mathrm{UE}\rightarrow\mathrm{S},n}^{[i,k]}},\label{eqn:R2}
\end{eqnarray}
where $\mathrm{SINR}^{[i,k]}_{\mathrm{DL}_n}$ and $\mathrm{SINR}^{[i,k]}_{\mathrm{UL}_n}$ are the received signal-to-interference-plus-noise-ratios (SINRs) at UE $k$ and the BS in subcarrier $i$ in spatial subchannel $n\in\{1,\ldots, N_{\mathrm{S}}\}$, respectively.  Although the objective function is now a scalar function with respect  to the optimization variables, it is still non-convex.  To obtain a tractable resource allocation algorithm design, we propose the following objective function approximation. In particular, the end-to-end DL and UL SINRs on subcarrier $i$ for UE $k$ can be approximated, respectively,  as
\begin{eqnarray}\label{eqn:SINR_1}
\mathrm{SINR}^{[i,k]}_{\mathrm{DL}_n}\hspace*{-1mm} &\approx&\hspace*{-1mm}\mathrm{\overline{SINR}}^{[i,k]}_{\mathrm{DL}_n}\quad \mbox{and}\quad \mathrm{SINR}^{[i,k]}_{\mathrm{UL}_n}\hspace*{-1mm} \approx\mathrm{\overline{SINR}}^{[i,k]}_{\mathrm{UL}_n},\quad \mbox{where}\\
\mathrm{\overline{SINR}}^{[i,k]}_{\mathrm{DL}_n}\hspace*{-1mm}&=&\hspace*{-1mm}\frac{\gamma_{\mathrm{B}\rightarrow\mathrm{S},n}^{[i]}P_{\mathrm{B}\rightarrow\mathrm{S},n}^{[i,k]}
P_{\mathrm{S}\rightarrow\mathrm{UE},n}^{[i,k]}\gamma_{\mathrm{S}\rightarrow\mathrm{UE},n}^{[i,k]}}{\gamma_{\mathrm{B}\rightarrow\mathrm{S},n}^{[i]}P_{\mathrm{B}\rightarrow\mathrm{S},n}^{[i,k]}
+ P_{\mathrm{S}\rightarrow\mathrm{UE},n}^{[i,k]}\gamma_{\mathrm{S}\rightarrow\mathrm{UE},n}^{[i,k]}},\,\, \mathrm{\overline{SINR}}^{[i,k]}_{\mathrm{UL}_n}\hspace*{-1mm}=\hspace*{-1mm}
\notag\frac{\gamma_{\mathrm{S}\rightarrow\mathrm{B},n}^{[i]}P_{\mathrm{S}\rightarrow\mathrm{B},n}^{[i,k]}
P_{\mathrm{UE}\rightarrow\mathrm{S},n}^{[i,k]}\gamma_{\mathrm{UE}\rightarrow\mathrm{S},n}^{[i,k]}}{\gamma_{\mathrm{S}\rightarrow\mathrm{B},n}^{[i]}P_{\mathrm{S}\rightarrow\mathrm{B},n}^{[i,k]}
+ P_{\mathrm{UE}\rightarrow\mathrm{S},n}^{[i,k]}\gamma_{\mathrm{UE}\rightarrow\mathrm{S},n}^{[i,k]}}\label{eqn:SINR_2}.
\end{eqnarray}
We note that this approximation is asymptotically tight for high SNR\footnote{It is expected that the high SNR assumption holds for the considered system due the short distance communication between the SUDAS and the UEs, i.e., $P_{\mathrm{S}\rightarrow\mathrm{UE},n}^{[i,k]}\gamma_{\mathrm{S}\rightarrow\mathrm{UE},n}^{[i,k]},
P_{\mathrm{UE}\rightarrow\mathrm{S},n}^{[i,k]}\gamma_{\mathrm{UE}\rightarrow\mathrm{S},n}^{[i,k]} \gg 1$. } \cite{JR:kwan_AF_relay,JR:MIMO_HD_relay1}.

The next step is to tackle the non-convexity due to combinatorial constraints C9 and C10 in (\ref{eqn:cross-layer-formulation}). To this end, we adopt the time-sharing relaxation approach. In particular, we relax $ s^{[i,k]}_{\mathrm{DL}}$ and $s^{[i,k]}_{\mathrm{UL}}$ in constraints C9 and C10 such that they are non-negative real valued optimization variables bounded from above by $\alpha$ and $\beta$, respectively   \cite{JR:Time_sharing_wei_yu}, i.e., $0\le  s^{[i,k]}_{\mathrm{DL}}\le \alpha$ and  $0\le  s^{[i,k]}_{\mathrm{UL}}\le \beta$. It has been shown in \cite{JR:Time_sharing_wei_yu} that the time-sharing relaxation is asymptotically optimal for a sufficiently large number of subcarriers\footnote{The duality gap due to the time-sharing relaxation  is virtually zero for practical  numbers of subcarriers, e.g. $n_{\mathrm{F}}\geq 8$
\cite{CN:large_subcarriers}. }. Next, we define a set with four auxiliary optimization variables  ${\cal \widetilde P}=\{\widetilde P_{\mathrm{B}\rightarrow\mathrm{S},n}^{[i,k]},\widetilde P_{\mathrm{UE}\rightarrow\mathrm{S},n}^{[i,k]},$ $\widetilde P_{\mathrm{S}\rightarrow\mathrm{UE},n}^{[i,k]},\widetilde P_{\mathrm{S}\rightarrow\mathrm{B},n}^{[i,k]} \}$ and rewrite the transformed objective function in \eqref{eqn:inner_loop} as:
\begin{eqnarray}\vspace*{-2mm}
&&\hspace*{-6mm}{\cal U}_{\mathrm{Trans}}({\cal \widetilde P},{\cal
S})=\sum_{k=1}^K\sum_{i=1}^{n_{\mathrm{F}}}\sum_{n=1}^{N_{\mathrm{S}}} \Bigg\{s^{[i,k]}_{\mathrm{DL}}
\log_2\Big(1+\frac{\widetilde{\mathrm{{SINR}}}^{[i,k]}_{\mathrm{DL}_n}}{s^{[i,k]}_{\mathrm{DL}}}\Big)+
s^{[i,k]}_{\mathrm{UL}}\log_2\Big(1+\frac{\widetilde{\mathrm{{SINR}}}^{[i,k]}_{\mathrm{UL}_n}}{s^{[i,k]}_{\mathrm{UL}}}\Big)\Bigg\} \\
&&\hspace*{-6mm}-\eta_{\mathrm{eff}}\Big(P_{\mathrm{C}_{\mathrm{B}}}\hspace*{-1mm}+\hspace*{-1mm} M P_{\mathrm{C}_{\mathrm{SUDAC}}}\hspace*{-1mm}+\hspace*{-1mm} K P_{\mathrm{C}_{\mathrm{UE}}}\hspace*{-1mm}+\hspace*{-1mm}\sum_{k=1}^{K}\sum_{i=1}^{n_{\mathrm{F}}}\sum_{n=1}^{N_{\mathrm{S}}} \varepsilon_{\mathrm{B}}\widetilde P_{\mathrm{B}\rightarrow\mathrm{S},n}^{[i,k]} \hspace*{-1mm}+\hspace*{-1mm} \varepsilon_{\mathrm{S}}\widetilde  P_{\mathrm{S}\rightarrow\mathrm{UE},n}^{[i,k]}\hspace*{-1mm}+\hspace*{-1mm}\varepsilon_{k}\widetilde P_{\mathrm{UE}\rightarrow\mathrm{S},n}^{[i,k]}\hspace*{-1mm}+\hspace*{-1mm} \varepsilon_{\mathrm{S}}\widetilde  P_{\mathrm{S}\rightarrow\mathrm{B},n}^{[i,k]}\Big)\notag \end{eqnarray}
where $\widetilde{\mathrm{SINR}}^{[i,k]}_{\mathrm{DL}_n}=\overline{\mathrm{{SINR}}}^{[i,k]}_{\mathrm{DL}_n}\Big|_{\Phi}$ and  $\Phi=\Big\{\widetilde P_{\mathrm{B}\rightarrow\mathrm{S},n}^{[i,k]}= P_{\mathrm{B}\rightarrow\mathrm{S},n}^{[i,k]} s^{[i,k]}_{\mathrm{DL}},\widetilde P_{\mathrm{S}\rightarrow\mathrm{UE},n}^{[i,k]}= P_{\mathrm{S}\rightarrow\mathrm{UE},n}^{[i,k]} s^{[i,k]}_{\mathrm{DL}},$ $ \widetilde P_{\mathrm{UE}\rightarrow\mathrm{S},n}^{[i,k]}= P_{\mathrm{UE}\rightarrow\mathrm{S},n}^{[i,k]}s^{[i,k]}_{\mathrm{UL}}  ,\widetilde P_{\mathrm{S}\rightarrow\mathrm{B},n}^{[i,k]}=P_{\mathrm{S}\rightarrow\mathrm{B},n}^{[i,k]} s^{[i,k]}_{\mathrm{UL}}       \Big \}$.  We note that the new auxiliary optimization variables in $\cal \widetilde P$ represent the actual transmit energy under the time-sharing condition.  As a result, the combinatorial-constraint relaxed  problem can be written as:\vspace*{-3mm}
\begin{eqnarray}\label{eqn:cross-layer-formulation-transformed}
&&\hspace*{10mm}\underset{{{\cal\widetilde P},}{\cal S}}\maxo\ \
{\cal U}_{\mathrm{Trans}}({\cal \widetilde P},{\cal
S}) \notag\\
\mathrm{s.t.}&\mbox{C1:}& \sum_{k=1}^{K}\sum_{i=1}^{n_{\mathrm{F}}}\sum_{n=1}^{N_{\mathrm{S}}}\widetilde P_{\mathrm{B}\rightarrow\mathrm{S},n}^{[i,k]} \le P_\mathrm{T},\notag\hspace*{12mm}
\mbox{C2:} \sum_{k=1}^{K}\sum_{i=1}^{n_{\mathrm{F}}}\sum_{n=1}^{N_{\mathrm{S}}} \widetilde  P_{\mathrm{S}\rightarrow\mathrm{UE},n}^{[i,k]}\le M P_{\max}, \notag\\
&\mbox{C3:}& \sum_{i=1}^{n_{\mathrm{F}}}\sum_{n=1}^{N_{\mathrm{S}}}\widetilde P_{\mathrm{UE}\rightarrow\mathrm{S},n}^{[i,k]} \le P_{\max_k}, \forall k,\quad\mbox{C4:} \sum_{k=1}^{K}\sum_{i=1}^{n_{\mathrm{F}}}\sum_{n=1}^{N_{\mathrm{S}}} \widetilde  P_{\mathrm{S}\rightarrow\mathrm{B},n}^{[i,k]}\le P_{\max}^{\mathrm{UL}},\notag\\
&\mbox{C5 -- C8,}& \mbox{C9: }0\leq s^{[i,k]}_{\mathrm{DL}}
\leq\alpha, \forall i,k, \hspace*{12mm}\mbox{C10: }0\leq s^{[i,k]}_{\mathrm{UL}}
\leq\beta,  \forall i,k, \mbox{ C11, C12}.
\end{eqnarray}
Optimization problem (\ref{eqn:cross-layer-formulation-transformed}) is jointly concave with respect to the auxiliary optimization variables $\cal \widetilde P$ and $\cal  S$. We note that by solving  optimization problem (\ref{eqn:cross-layer-formulation-transformed})  for $\widetilde P_{\mathrm{B}\rightarrow\mathrm{S}}^{[i,k]}$, $\widetilde P_{\mathrm{S}\rightarrow\mathrm{UE},n}^{[i,k]}$, $\widetilde P_{\mathrm{UE}\rightarrow\mathrm{S},n}^{[i,k]}$, $\widetilde P_{\mathrm{S}\rightarrow\mathrm{B},n}^{[i,k]}$, $s^{[i,k]}_{\mathrm{DL}}$, and $s^{[i,k]}_{\mathrm{UL}}$, we can
recover the solution for $ P_{\mathrm{B}\rightarrow\mathrm{S},n}^{[i,k]}$, $P_{\mathrm{S}\rightarrow\mathrm{UE},n}^{[i,k]}$, $P_{\mathrm{UE}\rightarrow\mathrm{S},n}^{[i,k]}$, and  $ P_{\mathrm{S}\rightarrow\mathrm{B},n}^{[i,k]}$.  Thus, the solution of (\ref{eqn:cross-layer-formulation-transformed}) is asymptotically optimal with respect to (\ref{eqn:cross-layer-formulation}) for high SNR and a sufficiently large number of subcarriers.

\begin{table}\caption{Iterative Resource Allocation Algorithm for SUDAS Assisted Communication}\label{table:algorithm2}\vspace*{-12mm}
\begin{algorithm} [H]                 
\caption{Alternating Optimization}          
\label{alg1}                           
\begin{algorithmic} [1]  \label{algorithm2}
\STATE Initialize the maximum number of iterations $L_{\max}$ and a small constant $\kappa\rightarrow 0$
\STATE Set iteration index $l=0$ and initialize $\{P_{\mathrm{B}\rightarrow\mathrm{S},n}^{[i,k]}(l),P_{\mathrm{S}\rightarrow\mathrm{UE},n}^{[i,k]}(l),s^{[i,k]}_{\mathrm{DL}}(l)\},
\{P_{\mathrm{S}\rightarrow\mathrm{B},n}^{[i,k]}(l),P_{\mathrm{UE}\rightarrow\mathrm{S},n}^{[i,k]}(l),$ $s^{[i,k]}_{\mathrm{UL}}(l)\}$, $\alpha$, $\beta$, and  $l=l+1$

\REPEAT [Loop]
\STATE For  given $P_{\mathrm{S}\rightarrow\mathrm{UE},n}^{[i,k]}(l-1)$ and $\alpha(l)$, solve  (\ref{eqn:cross-layer-formulation-transformed})  for $P_{\mathrm{B}\rightarrow\mathrm{S},n}^{[i,k]}$ by using (\ref{eqn:P1}) which leads to  intermediate power allocation variables $P_{\mathrm{B}\rightarrow\mathrm{S},n}^{[i,k]'}$

\STATE For given  $P_{\mathrm{B}\rightarrow\mathrm{S},n}^{[i,k]'}$ and $\alpha(l)$,  solve  (\ref{eqn:cross-layer-formulation-transformed})  for $P_{\mathrm{S}\rightarrow\mathrm{UE},n}^{[i,k]}$  via equation (\ref{eqn:P2})  which leads to  intermediate power allocation variables $P_{\mathrm{S}\rightarrow\mathrm{UE},n}^{[i,k]'}$

\STATE Update the DL  subcarrier allocation policy via \eqref{eqn:sub_selection} with $P_{\mathrm{S}\rightarrow\mathrm{UE},n}^{[i,k]}(l-1)$,  $P_{\mathrm{B}\rightarrow\mathrm{S},n}^{[i,k]'}$, and $\alpha(l)$  to obtain the intermediate  DL  subcarrier allocation policy $s^{[i,k]'}_{\mathrm{DL}}(l)$

\STATE For given  $P_{\mathrm{S}\rightarrow\mathrm{B},n}^{[i,k]}(l-1)$ and $\beta(l)$,   solve  (\ref{eqn:cross-layer-formulation-transformed})  for $P_{\mathrm{UE}\rightarrow\mathrm{S},n}^{[i,k]}$   via equation (\ref{eqn:P3})  which leads to  intermediate power allocation variables $P_{\mathrm{UE}\rightarrow\mathrm{S},n}^{[i,k]'}$

\STATE For  given $P_{\mathrm{UE}\rightarrow\mathrm{S},n}^{[i,k]'}$ and $\beta(l)$, solve  (\ref{eqn:cross-layer-formulation-transformed})  for $P_{\mathrm{S}\rightarrow\mathrm{B},n}^{[i,k]}$ by using (\ref{eqn:P4}) which leads to  intermediate power allocation variables $P_{\mathrm{S}\rightarrow\mathrm{B},n}^{[i,k]'}$

\STATE Update the UL  subcarrier allocation policy via \eqref{eqn:sub_selection2} with $P_{\mathrm{S}\rightarrow\mathrm{B},n}^{[i,k]}(l-1)$, $P_{\mathrm{UE}\rightarrow\mathrm{S},n}^{[i,k]'}$, and $\beta(l)$  to obtain the intermediate  UL  subcarrier allocation policy $s^{[i,k]'}_{\mathrm{UL}}(l)$

\STATE Update $\alpha$ and $\beta$ via standard linear programming methods   to obtain  intermediate solutions of $\alpha'$ and $\beta'$

\IF{$\abs{P_{\mathrm{S}\rightarrow\mathrm{UE},n}^{[i,k]'}-P_{\mathrm{S}\rightarrow\mathrm{UE},n}^{[i,k]}(l-1)}\le \kappa$ , $\abs{P_{\mathrm{B}\rightarrow\mathrm{S},n}^{[i,k]'}-P_{\mathrm{B}\rightarrow\mathrm{S},n}^{[i,k]}(l-1)}\le \kappa$ , $\abs{s^{[i,k]'}_{\mathrm{DL}}-s^{[i,k]}_{\mathrm{DL}}(l-1)}\le \kappa$, \\
$\abs{P_{\mathrm{UE}\rightarrow\mathrm{S},n}^{[i,k]'}-P_{\mathrm{UE}\rightarrow\mathrm{S},n}^{[i,k]}(l-1)}\le \kappa$ , $\abs{P_{\mathrm{S}\rightarrow\mathrm{B},n}^{[i,k]'}-P_{\mathrm{S}\rightarrow\mathrm{B},n}^{[i,k]}(l-1)}\le \kappa$ , $\abs{s^{[i,k]'}_{\mathrm{UL}}-s^{[i,k]}_{\mathrm{UL}}(l-1)}\le \kappa$\\
$\abs{\alpha'-\alpha(l-1)}\le \kappa$, and $\abs{\beta'-\beta(l-1)}\le \kappa$}
\STATE
Convergence = \TRUE, \quad\textbf{return}
$\{P_{\mathrm{S}\rightarrow\mathrm{UE},n}^{[i,k]'},P_{\mathrm{B}\rightarrow\mathrm{S},n}^{[i,k]'},s^{[i,k]'}_{\mathrm{DL}},
P_{\mathrm{UE}\rightarrow\mathrm{S},n}^{[i,k]'},P_{\mathrm{S}\rightarrow\mathrm{B},n}^{[i,k]'},s^{[i,k]'}_{\mathrm{UL}},\alpha',\beta'\}$
\ELSE
\STATE Convergence = \FALSE, $P_{\mathrm{S}\rightarrow\mathrm{UE},n}^{[i,k]}(l)=P_{\mathrm{S}\rightarrow\mathrm{UE},n}^{[i,k]'},
P_{\mathrm{B}\rightarrow\mathrm{S},n}^{[i,k]}(l)=P_{\mathrm{B}\rightarrow\mathrm{S},n}^{[i,k]'},
s^{[i,k]}_{\mathrm{DL}}(l)=s^{[i,k]'}_{\mathrm{DL}},
P_{\mathrm{UE}\rightarrow\mathrm{S},n}^{[i,k]}(l)=P_{\mathrm{UE}\rightarrow\mathrm{S},n}^{[i,k]'},
P_{\mathrm{S}\rightarrow\mathrm{B},n}^{[i,k]}(l)=P_{\mathrm{S}\rightarrow\mathrm{B},n}^{[i,k]'},
s^{[i,k]}_{\mathrm{UL}}(l)=s^{[i,k]'}_{\mathrm{UL}},\alpha(l)=\alpha',\beta(l)=\beta'$, $l=l+1$
 \ENDIF
\UNTIL{ $l=L_{\max}$}
\end{algorithmic}
\end{algorithm}\vspace*{-18mm}
\end{table}

Now, we propose an algorithm for solving the transformed problem in (\ref{eqn:cross-layer-formulation-transformed}). Although the transformed problem is jointly concave with respect to the optimization variables and can be solved by  convex programming solvers, it is difficult to obtain system design insight from a numerical solution. This  motivates us to design an iterative resource allocation algorithm which reveals the structure of energy-efficient resource allocation solutions and serves as a building block for the suboptimal algorithm proposed in the next section. The proposed iterative resource allocation algorithm is based on alternating  optimization. The algorithm is summarized in Table \ref{table:algorithm2} and is implemented by a repeated loop.  In line 2,  we first set the iteration index $l$ to zero
and initialize the resource allocation policy. Variables $P_{\mathrm{B}\rightarrow\mathrm{S},n}^{[i,k]}(l),P_{\mathrm{S}\rightarrow\mathrm{UE},n}^{[i,k]}(l),
P_{\mathrm{S}\rightarrow\mathrm{B},n}^{[i,k]}(l),P_{\mathrm{UE}\rightarrow\mathrm{S},n}^{[i,k]}(l)$, $s^{[i,k]}_{\mathrm{DL}}(l)$, $s^{[i,k]}_{\mathrm{UL}}(l)$, $\alpha(l)$, and $\beta(l)$ denote the resource allocation policy in the $l$-th iteration. Then, in each iteration, we solve  (\ref{eqn:cross-layer-formulation-transformed}), which leads to
\eqref{eqn:P1}--\eqref{eqn:P4}:
   \begin{eqnarray}\label{eqn:P1}
\hspace*{-4mm}P_{\mathrm{B}\rightarrow\mathrm{S},n}^{[i,k]}\hspace*{-3mm}&=&\hspace*{-2mm}\Bigg[
   \frac{\gamma_{\mathrm{S}\rightarrow\mathrm{UE},n}^{[i,k]} P_{\mathrm{S}\rightarrow\mathrm{UE},n}^{[i,k]} \Big(\Omega_{\mathrm{B}\rightarrow\mathrm{S},n}^{[i,k]}-\hspace*{-0.5mm}\gamma_{\mathrm{S}\rightarrow\mathrm{UE},n}^{[i,k]}
   P_{\mathrm{S}\rightarrow\mathrm{UE},n}^{[i,k]}\hspace*{-0.5mm}-\hspace*{-0.5mm}2\Big)}{2 (\gamma_{\mathrm{B}\rightarrow\mathrm{S},n}^{[i]} \gamma_{\mathrm{S}\rightarrow\mathrm{UE},n}^{[i,k]}
   P_{\mathrm{S}\rightarrow\mathrm{UE},n}^{[i,k]}\hspace*{-0.5mm}+\hspace*{-0.5mm}
   \gamma_{\mathrm{B}\rightarrow\mathrm{S},n}^{[i]})}\Bigg]^+,  \\
    \hspace*{-6mm}\Omega_{\mathrm{B}\rightarrow\mathrm{S},n}^{[i,k]} \hspace*{-3mm}&=&\hspace*{-3mm} \frac{\sqrt{\hspace*{-0.5mm}4 (\hspace*{-0.5mm}1\hspace*{-0.5mm}+\hspace*{-0.5mm}w^{[k]}_{\mathrm{DL}}\hspace*{-0.5mm}) \gamma_{\mathrm{B}\rightarrow\mathrm{S},n}^{[i]}(\hspace*{-0.5mm}1\hspace*{-0.5mm}+\hspace*{-0.5mm}
   \gamma_{\mathrm{S}\rightarrow\mathrm{UE},n}^{[i,k]} P_{\mathrm{S}\rightarrow\mathrm{UE},n}^{[i,k]}\hspace*{-0.5mm})\hspace*{-0.5mm}\hspace*{-0.5mm}+\hspace*{-0.5mm}\hspace*{-0.5mm}
   (\hspace*{-0.5mm}\gamma_{\mathrm{S}\rightarrow\mathrm{UE},n}^{[i,k]}\hspace*{-0.5mm})^2 (\hspace*{-0.5mm}\lambda \hspace*{-0.5mm}+\hspace*{-0.5mm}\eta_{\mathrm{eff}}\varepsilon_{\mathrm{B}}\hspace*{-0.5mm})
   (\hspace*{-0.5mm}P_{\mathrm{S}\rightarrow\mathrm{UE},n}^{[i,k]}\hspace*{-0.5mm})^2 \ln(\hspace*{-0.5mm}2\hspace*{-0.5mm})}}{\sqrt{\hspace*{-0.5mm}\lambda \hspace*{-0.5mm}+\hspace*{-0.5mm}\eta_{\mathrm{eff}}\varepsilon_{\mathrm{B}} } \sqrt{\hspace*{-0.5mm}\ln
   (\hspace*{-0.5mm}2\hspace*{-0.5mm})}},\label{eqn:omega1}
   \end{eqnarray}
with  $P_{\mathrm{S}\rightarrow\mathrm{UE},n}^{[i,k]}(l-1)$ and $\alpha(l)$ from the last iteration, where  $[x]^+=\max\{x,0\}$. Then, the obtained intermediate power allocation variable $P_{\mathrm{B}\rightarrow\mathrm{S},n}^{[i,k]}$  is used as an input for solving (\ref{eqn:cross-layer-formulation-transformed}) for $P_{\mathrm{S}\rightarrow\mathrm{UE},n}^{[i,k]}$  via the following equations:
\begin{eqnarray}
 \label{eqn:P2}\hspace*{-8mm}  P_{\mathrm{S}\rightarrow\mathrm{UE},n}^{[i,k]}\hspace*{-3mm}&=& \hspace*{-2mm}\Bigg[\frac{\gamma_{\mathrm{B}\rightarrow\mathrm{S},n}^{[i]} P_{\mathrm{B}\rightarrow\mathrm{S},n}^{[i,k]} \left( \Omega_{\mathrm{S}\rightarrow\mathrm{UE},n}^{[i,k]}-\gamma_{\mathrm{B}\rightarrow\mathrm{S},n}^{[i]} P_{\mathrm{B}\rightarrow\mathrm{S},n}^{[i,k]}-2\right)}{2 (\gamma_{\mathrm{B}\rightarrow\mathrm{S},n}^{[i]} \gamma_{\mathrm{S}\rightarrow\mathrm{UE},n}^{[i,k]} P_{\mathrm{B}\rightarrow\mathrm{S},n}^{[i,k]}+\gamma_{\mathrm{S}\rightarrow\mathrm{UE},n}^{[i,k]})}\Bigg]^+,\\
  \hspace*{-8mm}\Omega_{\mathrm{S}\rightarrow\mathrm{UE},n}^{[i,k]}\hspace*{-3mm}&=&\hspace*{-3mm}    \frac{\sqrt{\hspace*{-0.5mm}(\gamma_{\mathrm{B}\rightarrow\mathrm{S},n}^{[i]})^2 (\delta \hspace*{-0.5mm}+\hspace*{-0.5mm}\eta_{\mathrm{eff}}\varepsilon_{\mathrm{S}} ) (P_{\mathrm{B}\rightarrow\mathrm{S},n}^{[i,k]})^2 \ln(2)
 \hspace*{-0.5mm}+\hspace*{-0.5mm} (\gamma_{\mathrm{B}\rightarrow\mathrm{S},n}^{[i]} P_{\mathrm{B}\rightarrow\mathrm{S},n}^{[i,k]}\hspace*{-0.5mm}+\hspace*{-0.5mm}1) 4(1\hspace*{-0.5mm}+\hspace*{-0.5mm}w^{[k]}_{\mathrm{DL}}) \gamma_{\mathrm{S}\rightarrow\mathrm{UE},n}^{[i,k]}}}{\sqrt{\hspace*{-0.5mm}\delta \hspace*{-0.5mm}+\hspace*{-0.5mm}\eta_{\mathrm{eff}}\varepsilon_{\mathrm{S}} }
   \sqrt{\hspace*{-0.5mm}\ln (2)}}.\label{eqn:omega2}
\end{eqnarray}
\eqref{eqn:P1}--\eqref{eqn:P4} are obtained by standard convex optimization techniques. $\lambda$ and $\delta$  in \eqref{eqn:P1} and \eqref{eqn:P2}  are the Lagrange multipliers  for constraints C1 and C2 in (\ref{eqn:cross-layer-formulation-transformed}), respectively. In particular, $\lambda$ and $\delta$  are monotonically decreasing with respect to  $P_{\mathrm{B}\rightarrow\mathrm{S},n}^{[i,k]}$ and $ P_{\mathrm{S}\rightarrow\mathrm{UE},n}^{[i,k]}$, respectively, and control the transmit power at the BS and the SUDAS to satisfy constraints C1 and C2, respectively.    Besides, $w^{[k]}_{\mathrm{DL}}\geq 0$ is the Lagrange multiplier associated with the minimum required DL data rate constraint C5 for delay sensitive UE $k$.  The optimal values of $\lambda$, $\delta$, and $w^{[k]}_{\mathrm{DL}}$  in each iteration can be found by a standard gradient algorithm such that constraints C1,  C2, and C4  in (\ref{eqn:cross-layer-formulation-transformed}) are satisfied.  Variable
$\eta_{\mathrm{eff}}\ge 0$ generated by the Dinkelbach method prevents unnecessary energy expenditures by reducing the values of $\Omega_{\mathrm{B}\rightarrow\mathrm{S},n}^{[i,k]}$ and $\Omega_{\mathrm{S}\rightarrow\mathrm{UE},n}^{[i,k]}$ in \eqref{eqn:omega1} and \eqref{eqn:omega2}, respectively. Besides,   the power allocation strategy in \eqref{eqn:P1} and \eqref{eqn:P2} is  analogous to the  water-filling solution in traditional single-hop communication systems. In particular, $\Omega_{\mathrm{S}\rightarrow\mathrm{UE},n}^{[i,k]}$ and $\Omega_{\mathrm{B}\rightarrow\mathrm{S},n}^{[i,k]}$ act as water levels for controlling the allocated power. Interestingly, the water level in the power allocation for the BS-to-SUDAS link depends on the associated channel gain  which is different from the  power allocation in non-SUDAS assisted communication \cite{JR:virtual_MIMO2,CN:DAS_OFDMA}. Furthermore,  it can be seen from \eqref{eqn:omega1} and \eqref{eqn:omega2} that the water levels of
different users can be different. Specifically, if the end-to-end channel gains of two users are the same, to satisfy the data rate  requirement, the water level of the delay-sensitive
user is generally higher than
that of the non-delay sensitive
user.

After obtaining the intermediate DL power allocation policy, cf. lines 4, 5, we update the DL subcarrier allocation, cf. line 6, as:
  \begin{eqnarray}\hspace*{-3.5mm}
\label{eqn:sub_selection}s^{[i,k]}_{\mathrm{DL}}\hspace*{-2.0mm}&=&\hspace*{-2.0mm}
 \left\{ \begin{array}{rl}
 \alpha &\mbox{if $k =\arg\underset{t\in\{1,\ldots, K\}}\max \,(1+w^{[t]}_{\mathrm{DL}})\Big(\hspace*{-0.5mm}\sum_{n=1}^N\hspace*{-0.5mm}
\log_2\Big(1+\mathrm{\overline{SINR}}^{[i,t]}_{\mathrm{DL}_n}\Big)
\hspace*{-0.5mm}-\hspace*{-0.5mm}
\frac{\mathrm{\overline{SINR}}^{[i,t]}_{\mathrm{DL}_n}}{1+\mathrm{\overline{SINR}}^{[i,t]}_{\mathrm{DL}_n}}\hspace*{-0.5mm}\Big)$ }  \\
 0 &\mbox{ otherwise}
       \end{array} \right..
       \label{eqn:subcarrier_allocation}
\end{eqnarray}Here,  $\mathrm{\overline{SINR}}^{[i,k]}_{\mathrm{DL}_n}$ is obtained by substituting the intermediate solution of  $P_{\mathrm{B}\rightarrow\mathrm{S},n}^{[i,k]'}$ and $P_{\mathrm{S}\rightarrow\mathrm{UE},n}^{[i,k]'}$, i.e., (\ref{eqn:P1}) and (\ref{eqn:P2}), into (\ref{eqn:SINR_1}) in the $l$-th iteration. We note that the optimal value of $s^{[i,k]}_{\mathrm{DL}}$ of the relaxed problem
is a discrete value, cf. \eqref{eqn:sub_selection}, i.e., the constraint relaxation is tight.

Similarly, we optimize the UL power allocation variables,  $P_{\mathrm{UE}\rightarrow\mathrm{S},n}^{[i,k]}$ and $P_{\mathrm{S}\rightarrow\mathrm{B},n}^{[i,k]}$,  sequentially, cf. lines 7, 8, via the following equations:
 \begin{eqnarray}
  \label{eqn:P3}\hspace*{-4mm}  P_{\mathrm{UE}\rightarrow\mathrm{S},n}^{[i,k]}\hspace*{-3mm}&=& \hspace*{-2mm}\Bigg[\frac{\gamma_{\mathrm{S}\rightarrow\mathrm{B},n}^{[i]} P_{\mathrm{S}\rightarrow\mathrm{B},n}^{[i,k]} \left(\Omega_{\mathrm{UE}\rightarrow\mathrm{S},n}^{[i,k]}-\gamma_{\mathrm{S}\rightarrow\mathrm{B},n}^{[i]} P_{\mathrm{S}\rightarrow\mathrm{B},n}^{[i,k]}-2\right)}{2 (\gamma_{\mathrm{S}\rightarrow\mathrm{B},n}^{[i]} \gamma_{\mathrm{UE}\rightarrow\mathrm{S},n}^{[i,k]} P_{\mathrm{S}\rightarrow\mathrm{B},n}^{[i,k]}+\gamma_{\mathrm{UE}\rightarrow\mathrm{S},n}^{[i,k]})}\Bigg]^+,
  \\
   \label{eqn:P4}
\hspace*{-4mm}P_{\mathrm{S}\rightarrow\mathrm{B},n}^{[i,k]}\hspace*{-3mm}&=&\hspace*{-2mm}\Bigg[
   \frac{\gamma_{\mathrm{UE}\rightarrow\mathrm{S},n}^{[i,k]} P_{\mathrm{UE}\rightarrow\mathrm{S},n}^{[i,k]} \left(\hspace*{-0.5mm}\Omega_{\mathrm{S}\rightarrow\mathrm{B},n}^{[i,k]}\hspace*{-0.5mm}-\hspace*{-0.5mm}\gamma_{\mathrm{UE}\rightarrow\mathrm{S},n}^{[i,k]}P_{\mathrm{UE}\rightarrow\mathrm{S},n}^{[i,k]}\hspace*{-0.5mm}-\hspace*{-0.5mm}2\right)}{2 (\gamma_{\mathrm{S}\rightarrow\mathrm{B},n}^{[i]} \gamma_{\mathrm{UE}\rightarrow\mathrm{S},n}^{[i,k]}
   P_{\mathrm{UE}\rightarrow\mathrm{S},n}^{[i,k]}\hspace*{-0.5mm}+\hspace*{-0.5mm}\gamma_{\mathrm{S}\rightarrow\mathrm{B},n}^{[i]})}\Bigg]^+,
   \end{eqnarray}
   respectively, where
    \begin{eqnarray}\label{eqn:omega3}
   \hspace*{-6mm}\Omega_{\mathrm{UE}\rightarrow\mathrm{S},n}^{[i,k]}\hspace*{-3mm}&=&\hspace*{-3mm}
   \frac{\hspace*{-1.5mm}\sqrt{\hspace*{-0.5mm}(\gamma_{\mathrm{S}\rightarrow\mathrm{B},n}^{[i]})^2 ( \psi_k\hspace*{-0.5mm}+\hspace*{-0.5mm}\eta_{\mathrm{eff}}\varepsilon_{k}) (P_{\mathrm{S}\rightarrow\mathrm{B},n}^{[i,k]})^2 \ln(2)
 \hspace*{-0.5mm}+\hspace*{-0.5mm} (\gamma_{\mathrm{S}\rightarrow\mathrm{B},n}^{[i]} P_{\mathrm{S}\rightarrow\mathrm{B},n}^{[i,k]}\hspace*{-0.5mm}+\hspace*{-0.5mm}1) 4(1\hspace*{-0.5mm}+\hspace*{-0.5mm}w^{[k]}_{\mathrm{UL}}) \gamma_{\mathrm{UE}\rightarrow\mathrm{S},n}^{[i,k]}}}{\sqrt{\hspace*{-0.5mm} \psi_k\hspace*{-0.5mm}+\hspace*{-0.5mm}\eta_{\mathrm{eff}}\varepsilon_{k}  }
   \sqrt{\hspace*{-0.5mm}\ln (2)}},\\
  \hspace*{-6mm}\Omega_{\mathrm{S}\rightarrow\mathrm{B},n}^{[i,k]}\hspace*{-3mm}&=&\hspace*{-3mm}  \frac{\hspace*{-1.5mm}\sqrt{\hspace*{-0.5mm}4 (\hspace*{-0.5mm}1\hspace*{-0.5mm}+\hspace*{-0.5mm}w^{[k]}_{\mathrm{UL}}\hspace*{-0.5mm}) \gamma_{\mathrm{S}\rightarrow\mathrm{B},n}^{[i]}(\hspace*{-0.5mm}1\hspace*{-0.5mm}+\hspace*{-0.5mm}
   \gamma_{\mathrm{UE}\rightarrow\mathrm{S},n}^{[i,k]} P_{\mathrm{UE}\rightarrow\mathrm{S},n}^{[i,k]}\hspace*{-0.5mm})\hspace*{-0.5mm}\hspace*{-0.5mm}+\hspace*{-0.5mm}\hspace*{-0.5mm}(\hspace*{-0.5mm}\gamma_{\mathrm{UE}\rightarrow\mathrm{S},n}^{[i,k]}\hspace*{-0.5mm})^2 (\hspace*{-0.5mm}\phi\hspace*{-0.5mm}+\hspace*{-0.5mm}\eta_{\mathrm{eff}}\varepsilon_{\mathrm{S}}\hspace*{-0.5mm})
   (\hspace*{-0.5mm}P_{\mathrm{UE}\rightarrow\mathrm{S},n}^{[i,k]}\hspace*{-0.5mm})^2 \ln(\hspace*{-0.5mm}2\hspace*{-0.5mm})}}{\sqrt{\hspace*{-0.5mm}  \phi\hspace*{-0.5mm}+\hspace*{-0.5mm}\eta_{\mathrm{eff}}\varepsilon_{\mathrm{S}}} \sqrt{\hspace*{-0.5mm}\ln
   (\hspace*{-0.5mm}2\hspace*{-0.5mm})}}.\label{eqn:omega4}
   \end{eqnarray}
    $\psi_k$ and $\phi$  in \eqref{eqn:omega3} and \eqref{eqn:omega4}  are the Lagrange multipliers  with  respect to power consumption constraints C3 and C4 in (\ref{eqn:cross-layer-formulation-transformed}), respectively. Besides, $w^{[k]}_{\mathrm{UL}}$ is the Lagrange multiplier associated with the minimum required UL data rate constraint C6 for delay sensitive UE $k$.  The optimal values of $\psi_k$, $\phi$, and $w^{[k]}_{\mathrm{UL}}$  in each iteration can be easily obtained again with a standard gradient algorithm such that constraints C3,  C4, and C6  in (\ref{eqn:cross-layer-formulation-transformed}) are satisfied.

Then, we update the UL subcarrier allocation policy $s^{[i,k]}_{\mathrm{UL}}$ via\vspace*{-2mm}
\begin{eqnarray}\hspace*{-3.5mm}
       \label{eqn:sub_selection2}s^{[i,k]}_{\mathrm{UL}}\hspace*{-2.0mm}&=&\hspace*{-2.0mm}
 \left\{ \begin{array}{rl}
 \beta &\mbox{if $k =\arg\underset{t\in\{1,\ldots, K\}}\max \,(1+w^{[t]}_{\mathrm{UL}})\Big(\hspace*{-0.5mm}\sum_{n=1}^N\hspace*{-0.5mm}
\log_2\Big(1+\mathrm{\overline{SINR}}^{[i,t]}_{\mathrm{UL}_n}\Big)
\hspace*{-0.5mm}-\hspace*{-0.5mm}
\frac{\mathrm{\overline{SINR}}^{[i,t]}_{\mathrm{UL}_n}}{1+\mathrm{\overline{SINR}}^{[i,t]}_{\mathrm{UL}_n}}\hspace*{-0.5mm}\Big)$ }  \\
 0 &\mbox{ otherwise}
       \end{array} \right..
\end{eqnarray}
 where $\mathrm{\overline{SINR}}^{[i,k]}_{\mathrm{UL}_n}$ is obtained by substituting the intermediate solutions for  $P_{\mathrm{S}\rightarrow\mathrm{B},n}^{[i,k]'}$ and $ P_{\mathrm{UE}\rightarrow\mathrm{S},n}^{[i,k]'}$, i.e., (\ref{eqn:P3}) and (\ref{eqn:P4}), into (\ref{eqn:SINR_1}) in the $l$-th iteration. Again, the  constraint relaxation is tight.

  Subsequently, for a given UL and DL power allocation policy and given $s^{[i,k]}_{\mathrm{DL}}$ and $s^{[i,k]}_{\mathrm{UL}}$, the optimization problem is a linear programming with respect to $\alpha$ and $\beta$. Thus, we can update $\alpha$ and $\beta$ via standard linear programming methods  to obtain  intermediate solutions for $\alpha'$ and $\beta'$.  Then, the overall procedure is repeated iteratively until we reach the maximum number of iterations or convergence is achieved.  We note that for a sufficient number of iterations,  the convergence to the optimal solution of (\ref{eqn:cross-layer-formulation-transformed}) is guaranteed  since  (\ref{eqn:cross-layer-formulation-transformed}) is jointly  concave with respect to the optimization variables \cite{JR:AO}. Besides, the proposed algorithm has a polynomial time computational complexity.

\vspace*{-4mm}
\subsection{Suboptimal Solution}\vspace*{-2mm}
In the last section, we proposed an asymptotically globally optimal algorithm based on the high SNR assumption.
In this section, we propose a suboptimal resource allocation algorithm which achieves a local optimal solution of \eqref{eqn:cross-layer-formulation} for arbitrary SNR values. Similar to the  asymptotically optimal solution,
 we apply Theorem \ref{Thm:1}, Theorem \ref{Thm:Diagonalization_optimal}, and Algorithm \ref{algorithm1} to simplify the   power allocation and subcarrier allocation. In particular, the DL and UL data rates of UE $k$ on subcarrier $i$ are given by \eqref{eqn:R1} and \eqref{eqn:R2}, respectively.   It can  be observed that \eqref{eqn:R1} and \eqref{eqn:R2} are concave functions with respect to  $P_{\mathrm{B}\rightarrow\mathrm{S},n}^{[i,k]}$, $ P_{\mathrm{S}\rightarrow\mathrm{B},n}^{[i,k]}$, $P_{\mathrm{S}\rightarrow\mathrm{UE},n}^{[i,k]}$, and $P_{\mathrm{UE}\rightarrow\mathrm{S},n}^{[i,k]}$ individually, when the other variables are fixed.  Thus, we can apply alternating optimization to obtain a local optimal solution \cite{JR:AO} of \eqref{eqn:cross-layer-formulation}. We note that unlike the proposed asymptotically optimal scheme, the high SNR assumption is not required to convexify the problem.  The suboptimal solution can be  obtained by Algorithm \ref{algorithm2} in Table \ref{table:algorithm2}, but now, we update the   power allocation variables, i.e., lines 4,\,5,\,7,\,8, in Algorithm 2,  by using the following equations:
 \begin{eqnarray}
  \label{eqn:sub_opt1}\hspace*{-8mm}  P_{\mathrm{B}\rightarrow\mathrm{S},n}^{[i,k]}\hspace*{-3mm}&=&\hspace*{-3mm} \frac{1}{\gamma_{\mathrm{B}\rightarrow\mathrm{S},n}^{[i]}} \Bigg[\frac{P_{\mathrm{S}\rightarrow\mathrm{UE},n}^{[i,k]}\gamma_{\mathrm{S}\rightarrow\mathrm{UE},n}^{[i,k]}}{2}\Bigg(
  \sqrt{1+\frac{4\gamma_{\mathrm{B}\rightarrow\mathrm{S},n}^{[i]}  (1+w^{[k]}_{\mathrm{DL}})}{ \gamma_{\mathrm{S}\rightarrow\mathrm{UE},n}^{[i,k]} P_{\mathrm{S}\rightarrow\mathrm{UE},n}^{[i,k]}\ln(2)(\hspace*{-0.5mm}\lambda \hspace*{-0.5mm}+\hspace*{-0.5mm}\eta_{\mathrm{eff}}\varepsilon_{\mathrm{B}})}}-1\Bigg)-1\Bigg]^+, \\
    \label{eqn:sub_opt1_1}\hspace*{-8mm}  P_{\mathrm{S}\rightarrow\mathrm{UE},n}^{[i,k]}\hspace*{-3mm}&=&\hspace*{-3mm} \frac{1}{\gamma_{\mathrm{S}\rightarrow\mathrm{UE},n}^{[i,k]}} \Bigg[\frac{P_{\mathrm{B}\rightarrow\mathrm{S},n}^{[i,k]}\gamma_{\mathrm{B}\rightarrow\mathrm{S},n}^{[i]}}{2}\Bigg(
  \sqrt{1+\frac{4\gamma_{\mathrm{S}\rightarrow\mathrm{UE},n}^{[i,k]}  (1+w^{[k]}_{\mathrm{DL}})}{ \gamma_{\mathrm{B}\rightarrow\mathrm{S},n}^{[i]} P_{\mathrm{B}\rightarrow\mathrm{S},n}^{[i,k]}\ln(2)(\hspace*{-0.5mm}\delta \hspace*{-0.5mm}+\hspace*{-0.5mm}\eta_{\mathrm{eff}}\varepsilon_{\mathrm{S}})}}-1\Bigg)-1\Bigg]^+,     \\
\label{eqn:sub_opt2}
\hspace*{-8mm}  P_{\mathrm{S}\rightarrow\mathrm{B},n}^{[i,k]}\hspace*{-3mm}&=& \hspace*{-3mm} \frac{1}{\gamma_{\mathrm{S}\rightarrow\mathrm{B},n}^{[i]}} \Bigg[\frac{P_{\mathrm{UE}\rightarrow\mathrm{S},n}^{[i,k]}\gamma_{\mathrm{UE}\rightarrow\mathrm{S},n}^{[i,k]}}{2}\Bigg(
  \sqrt{1+\frac{4\gamma_{\mathrm{S}\rightarrow\mathrm{B},n}^{[i]}  (1+w^{[k]}_{\mathrm{UL}})}{ \gamma_{\mathrm{UE}\rightarrow\mathrm{S},n}^{[i,k]} P_{\mathrm{UE}\rightarrow\mathrm{S},n}^{[i,k]}\ln(2)(\phi\hspace*{-0.5mm}+\hspace*{-0.5mm}
  \eta_{\mathrm{eff}}\varepsilon_{\mathrm{S}})}}-1\Bigg)-1\Bigg]^+,\\
  \hspace*{-8mm}  P_{\mathrm{UE}\rightarrow\mathrm{S},n}^{[i,k]}\hspace*{-3mm}&=& \hspace*{-3mm} \frac{1}{\gamma_{\mathrm{UE}\rightarrow\mathrm{S},n}^{[i,k]}} \Bigg[\frac{P_{\mathrm{S}\rightarrow\mathrm{B},n}^{[i,k]}\gamma_{\mathrm{S}\rightarrow\mathrm{B},n}^{[i]}}{2}\Bigg(
  \sqrt{1+\frac{4\gamma_{\mathrm{UE}\rightarrow\mathrm{S},n}^{[i,k]}  (1+w^{[k]}_{\mathrm{UL}})}{ \gamma_{\mathrm{S}\rightarrow\mathrm{B},n}^{[i]} P_{\mathrm{S}\rightarrow\mathrm{B},n}^{[i,k]}\ln(2)(\psi_k\hspace*{-0.5mm}+\hspace*{-0.5mm}\eta_{\mathrm{eff}}\varepsilon_{k})}}-1\Bigg)-1\Bigg]^+,
   \end{eqnarray}
   which are obtained by applying standard optimization technique \cite{JR:MIMO_HD_relay1}.   Besides, the subcarrier allocation policies for DL and UL are still given by
  \eqref{eqn:sub_selection} and \eqref{eqn:sub_selection2}, respectively,  except that we replace
$\mathrm{\overline{SINR}}^{[i,k]}_{\mathrm{DL}_n}, \mathrm{\overline{SINR}}^{[i,k]}_{\mathrm{UL}_n}$ by $\mathrm{{SINR}}^{[i,k]}_{\mathrm{DL}_n},\mathrm{{SINR}}^{[i,k]}_{\mathrm{UL}_n}$ in \eqref{eqn:SINR_1}, respectively. The optimization variables are updated repeatedly  until convergence or the maximum number of iterations is reached. In contrast to the asymptotically optimal algorithm in Section \ref{sect:optimal_solution}, which may not even achieve a locally optimal solution for finite SNRs,  the suboptimal iterative algorithm is  guaranteed to converge to a local optimum \cite{JR:AO} for arbitrary SNR values.
\vspace*{-3mm}
\section{Results and Discussion}\label{sect:result-discussion}\vspace*{-2mm}
In this section, we evaluate the system performance based on Monte Carlo simulations. We assume an indoor environment with  $K=4$ UEs and $M$ SUDACs and an outdoor BS. The distances between the BS and UEs and between each SUDAS and each UE are $100$ meters and $4$ meters, respectively.
 For the BS-to-SUDAS links, we adopt the Urban macro outdoor-to-indoor scenario of the Wireless World Initiative New Radio (WINNER+) channel model \cite{Spec:Winner+}. The center frequency  and the bandwidth of the licensed band are  $800$ MHz and $20$ MHz, respectively. There are $n_{\mathrm{F}}=1200$  subcarriers  with $15$ kHz  subcarrier bandwidth resulting in $18$ MHz signal bandwidth for data transmission\footnote{The proposed SUDAS can be easily extended to the case when carrier aggregation is implemented at the BS to create a large signal bandwidth ($\sim 100$ MHz) in the licensed band. }. Hence, the BS-to-SUDAS link configuration is in accordance with the system parameters adopted in the Long Term Evolution (LTE) standard \cite{Spec:LTE}.  As for the SUDAS-to-UE links, we adopt the IEEE $802.11$ad channel model \cite{Spec:60GHz} in the range of $60$ GHz and assume that  $M$ orthogonal sub-bands are available. The maximum transmit power of the SUDACs and UEs is set to $MP_{\max}=P_{\max}^{\mathrm{UL}}=P_{\max_k}=23$ dBm which is in accordance with the maximum power spectral density suggested by the Harmonized European Standard for the mmW frequency band, i.e.,  $13$ dBm-per-MHz, and the typical maximum transmit power budgets of UEs. For simplicity, we assume that $N_{\mathrm{S}}=\min\{N,M\}$ for
studying the system performance. We model the SUDAS-to-BS and UE-to-SUDAS links as the conjugate transpose of the BS-to-SUDAS and SUDAS-to-UE links, respectively. Also, the power amplifier efficiencies of all power amplifiers are set to $25\%$. The circuit power consumption for the BS and each BS antenna are given by $P_{\mathrm{C}_{\mathrm{B}}}=15$ Watt and $P_{\mathrm{Ant}_{\mathrm{B}}}=0.975$ Watt, respectively \cite{CN:static_power,CN:power_consumption_elements}. The circuit power consumption per SUDAC and  UE  are set to $P_{\mathrm{C}_{\mathrm{SUDAC}}}=0.1$ Watt \cite{CN:60GHZ_power} and $P_{\mathrm{C}_{\mathrm{UE}}} =1$ Watt, respectively. We assume that there is always one  delay sensitive UE requiring  $R^{\mathrm{DL}}_{\min_k}=20$ Mbit/s and $R^{\mathrm{UL}}_{\min_k}=20$ Mbit/s in DL and UL, respectively. Also, $N_{\mathrm{S}}$ is chosen as $N_{\mathrm{S}}=
 \min\{\Rank(\mathbf{H}^{[i,k]}_{\mathrm{S}\rightarrow\mathrm{UE}}),\Rank(\mathbf{H}^{[i]}_{\mathrm{B}\rightarrow\mathrm{S}})\}$. All results
were averaged over $10000$ different multipath fading channel realizations.

\vspace*{-6mm}
\subsection{Convergence of the Proposed Iterative Algorithm}\vspace*{-2mm}
Figure \ref{fig:convergence} illustrates the convergence of
the proposed optimal and suboptimal algorithms  for $N=8$ antennas at the BS and $M=8$ SUDACs, and different maximum transmit powers at the BS, $P_{\mathrm{T}}$. We compare the system performance of the proposed algorithms with a performance upper bound  which is obtained by computing the optimal objective value in (\ref{eqn:cross-layer-formulation-transformed}) for noise-free reception at the UEs and the BS. The performance gap between the asymptotically optimal performance and the upper bound constitutes an upper bound on the performance loss due to the high SINR approximation adopted in (\ref{eqn:SINR_1}). The number of iterations is defined as the aggregate number of iterations required by Algorithms $1$ and $2$.  It can be observed that
the proposed asymptotically optimal algorithm approaches $99\%$ of the upper bound
value after $20$ iterations which  confirms the practicality of the proposed iterative algorithm. Besides, the suboptimal resource allocation algorithm, achieves $90\%$ of the upper bound
value in the low transmit power regime, i.e., $P_{\mathrm{T}}=19$ dBm, and virtually the same energy efficiency as the upper bound performance in the high transmit power regime, i.e., $P_{\mathrm{T}}=46$ dBm.  In the following case studies, the number of iterations is set to $30$ in order to illustrate the performance of the proposed
algorithms.

\begin{figure}[t]\vspace*{-6mm}
 \centering
 \begin{minipage}[b]{0.45\linewidth} \hspace*{-1cm}
\includegraphics[width=3.6 in]{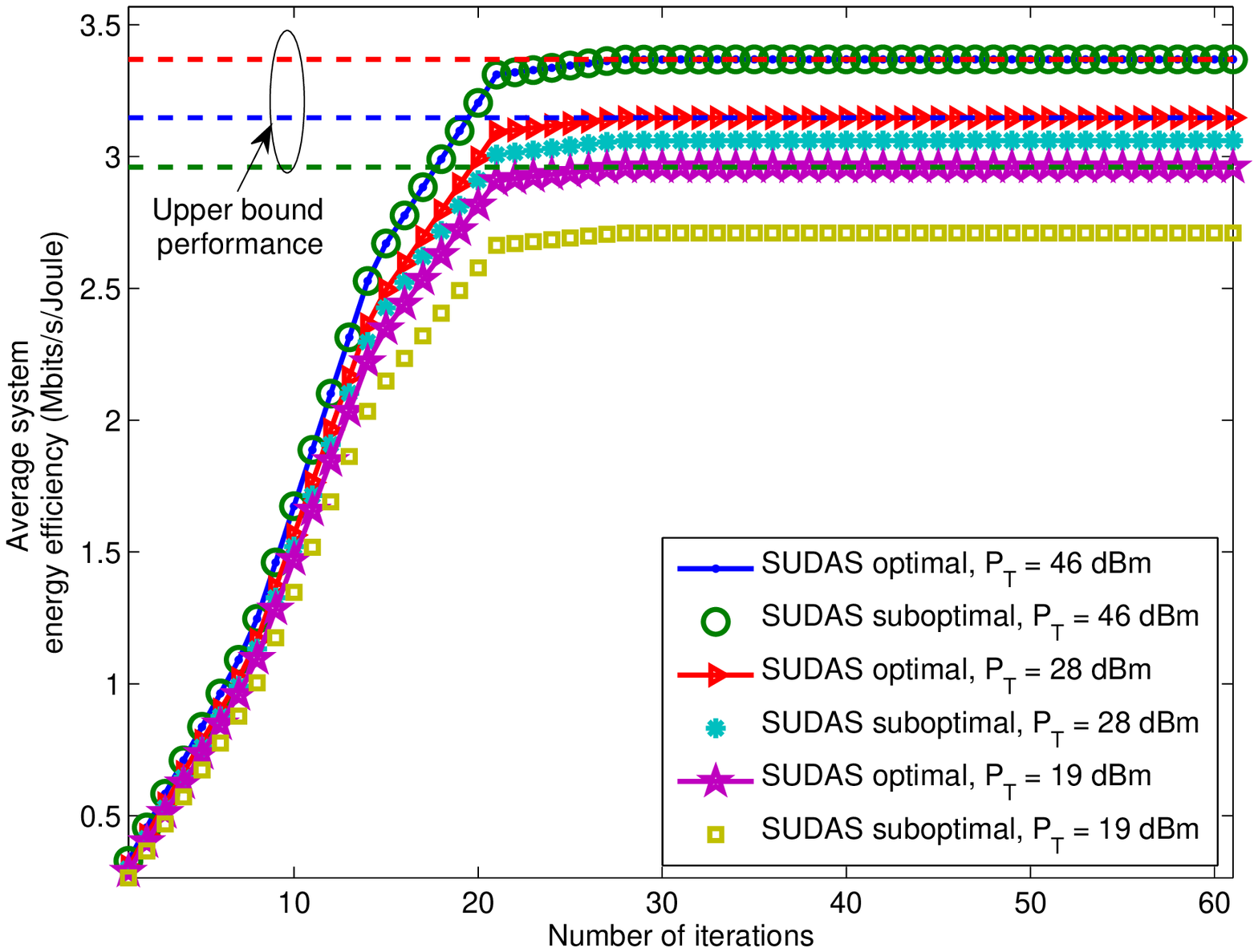}\vspace*{-8mm}
\caption{Average energy efficiency (Mbits/Joule) versus the number of iterations for different maximum transmit power budgets at the BS.  }\label{fig:convergence}
 \end{minipage}\hspace*{1.1cm}
 \begin{minipage}[b]{0.45\linewidth} \hspace*{-1cm}
\includegraphics[width=3.6 in]{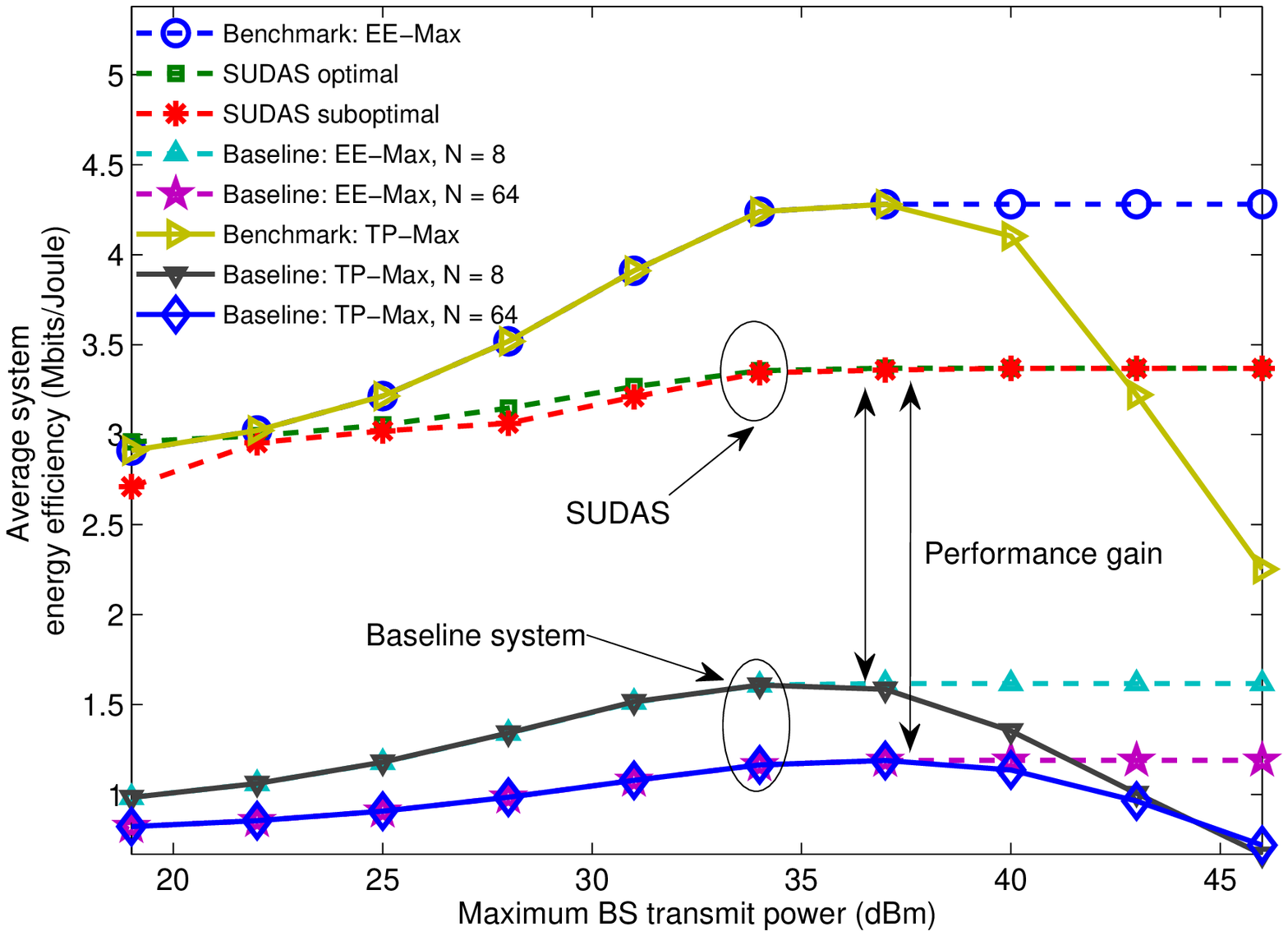}\vspace*{-8mm}
\caption{Average energy efficiency (Mbits/Joule) versus the maximum transmit power at the BS (dBm) for different communication systems. The  double-sided arrows
indicate  the performance gain achieved by the proposed SUDAS. } \label{fig:EE_PT}
 \end{minipage}\vspace*{-10mm}
\end{figure}
\vspace*{-6mm}
\subsection{Average System Energy Efficiency versus Maximum Transmit Power}
Figure \ref{fig:EE_PT} illustrates the average system
energy efficiency versus the maximum DL transmit power at the BS for $M=8$ SUDACs for different systems and $N=8$ BS antennas. It can be observed that the average system energy efficiency
of the two proposed resource allocation algorithms for SUDAS is a monotonically non-decreasing function of $P_{\mathrm{T}}$. In particular, starting
from a small value of $P_{\mathrm{T}}$, the energy efficiency  increases slowly with increasing $P_{\mathrm{T}}$ and
then saturates when $P_{\mathrm{T}}>37$ dBm. This is due to the fact that the two proposed algorithms strike a
balance between system energy efficiency and power consumption. In fact, once
the maximum energy efficiency of the SUDAS is achieved, even if there is more power available for transmission, the BS will not consume extra DL transmit power for improving the data rate, cf. \eqref{eqn:P1}.  This is because a further increase in the BS transmit power
would only result in a degradation of the energy efficiency.  Moreover, we compare the energy efficiency of the proposed SUDAS with a benchmark
MIMO system and a baseline system.  We focus on two system design objectives for the reference systems, namely, system throughput maximization (TP-Max) and energy efficiency maximization (EE-Max).
For the benchmark MIMO system, we assume that each UE is equipped with $N$ receive antennas but the SUDAS is not used and optimal resource allocation is performed\footnote{The optimal resource allocation for the benchmark system can be obtained by following a similar method as the one  proposed in this paper applying also fractional programming and majorization theory. }. Besides, we assume that the circuit power consumption at the UE does not scale with the number of antennas and only the licensed band is used for the benchmark system.  In other words, the average system energy efficiency of the benchmark system serves as a performance upper bound for the proposed SUDAS.  For the baseline system, we assume that the BS and the single-antenna UEs perform optimal resource allocation and utilize only the licensed frequency band, i.e., the SUDAS is not used. As can be observed from Figure \ref{fig:EE_PT},  for high BS transmit power budgets, the SUDAS achieves more than $80\%$ of the benchmark MIMO system performance  even though the UEs are only equipped with single antennas. Also, the SUDAS provides a huge system performance gain compared to the baseline system which does not employ SUDAS since the proposed  SUDAS allows the single-antenna UEs to exploit spatial and frequency multiplexing gains. On the other hand,  increasing the number of BS antennas  dramatically in the baseline system from $N=8$ to  $N=64$, i.e., to a large-scale antenna system, does not necessarily improve the system energy efficiency. In fact, in the baseline system, the higher power consumption, which increases linearly with the number of BS antennas, outweighs the system throughput gain, which only scales logarithmically with  the additional BS antennas.

Figure \ref{fig:time_allocation} depicts the average time allocation for DL and UL transmission. It can be observed that the optimal time allocation depends on the  transmit power budget of the systems.   In particular, when the power budget of the BS for DL communication is small compared to the total transmit power budget for UL communication, e.g. $P_{\mathrm{T}}\le 28$ dBm, the  period of time allocated for DL transmission is shorter than that allocated for UL transmission. Because of the limited power budget and the circuit power consumption, it is preferable for the BS to transmit a sufficiently large power over a short period of time rather than a small power over a longer time  to maximize the system energy efficiency and to fulfill the data rate requirement of the DL delay sensitive UEs. On the contrary, when the power budget of the BS is  large compared to that of the UEs,  the system allocates more time resources to the DL compared to the UL, since the  BS can now transmit a large enough  power to compensate the circuit power consumption for a longer time span to maximize the system energy efficiency. \vspace*{-5mm}
\subsection{Average System Throughput versus Maximum Transmit Power}\vspace*{-2mm}
\begin{figure}[t]\vspace*{-6mm}
 \centering
 \begin{minipage}[b]{0.45\linewidth} \hspace*{-1cm}
\includegraphics[width=3.5 in]{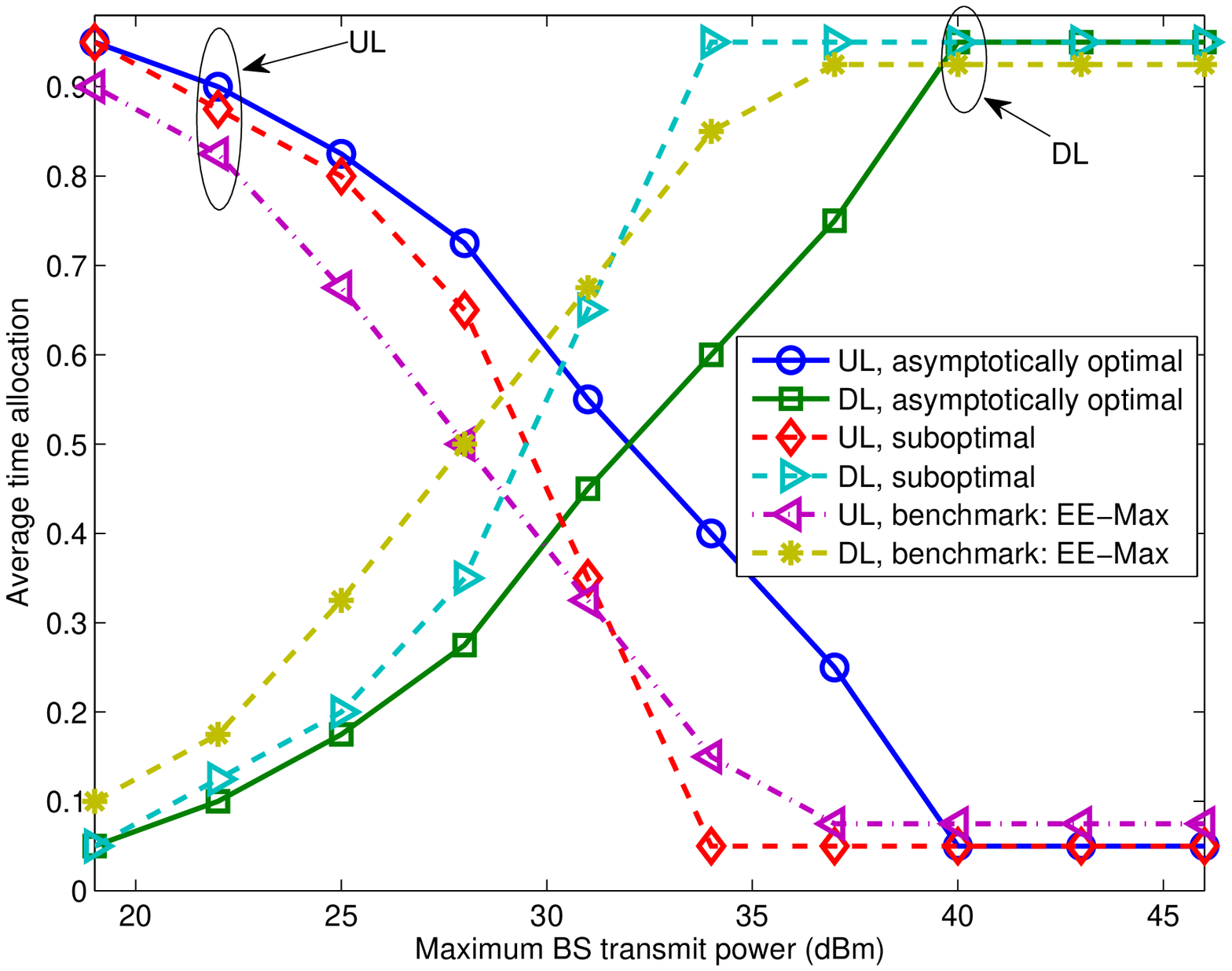}\vspace*{-8mm}
\caption{Average DL and UL transmission durations versus the maximum transmit power at the BS (dBm).  }\label{fig:time_allocation}
 \end{minipage}\hspace*{1.1cm}
 \begin{minipage}[b]{0.45\linewidth} \hspace*{-1cm}
\includegraphics[width=3.5 in]{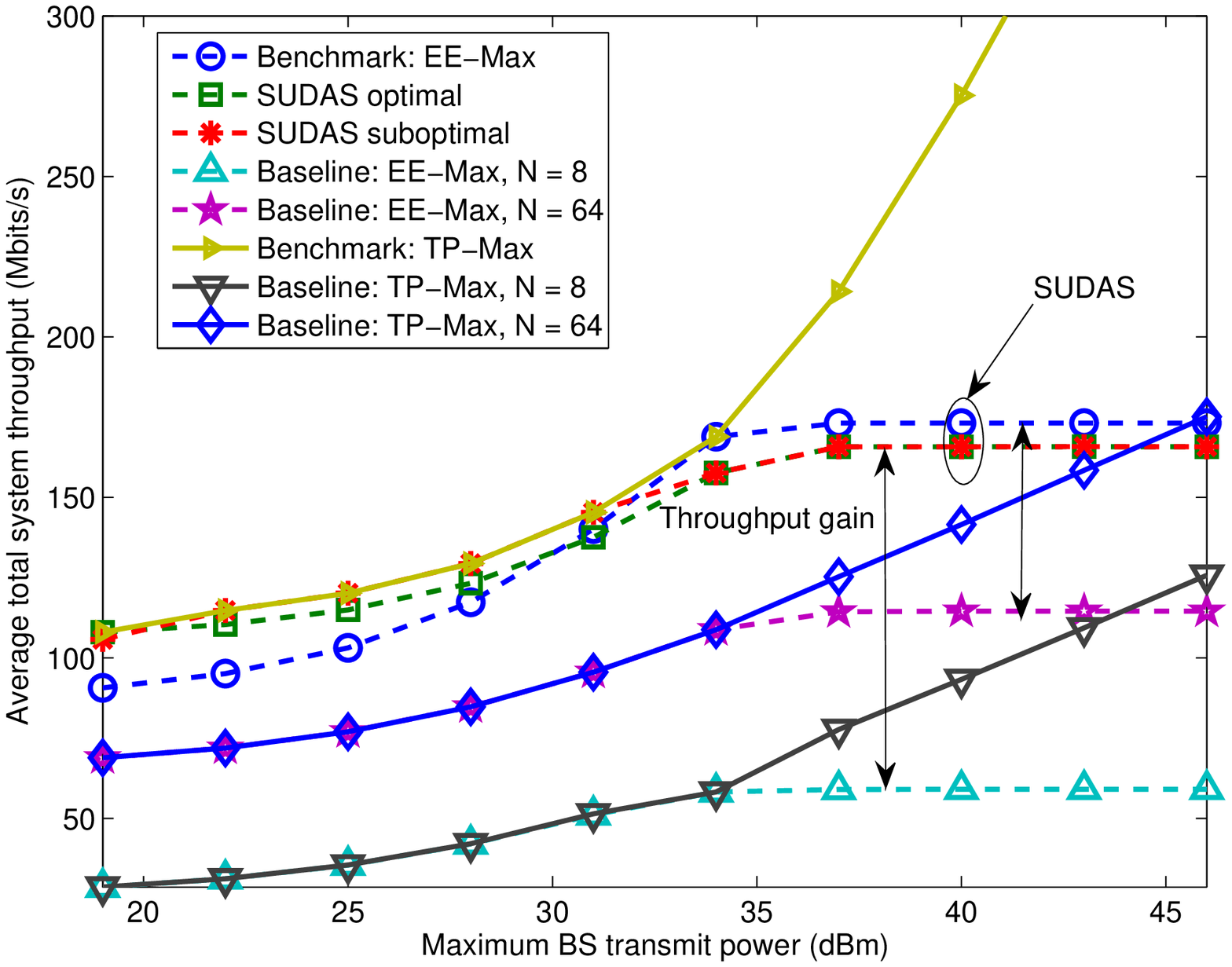}\vspace*{-8mm}
\caption{Average system throughput (Mbits/s) versus the maximum transmit power at the BS (dBm). } \label{fig:CAP_TP}
 \end{minipage}\vspace*{-8mm}
\end{figure}
Figure \ref{fig:CAP_TP}   illustrates the average system
throughput versus the maximum transmit power at the BS for $N=8$ BS antennas,  $K=4$ UEs, and $M=8$.
We compare the two proposed algorithms with
the two aforementioned reference systems. The proposed SUDAS performs closely to the benchmark scheme   in the low DL transmit power budget regime, e.g. $P_{\mathrm{T}}\leq 31$ dBm. This is due to the fact that the proposed SUDAS allows the single-antenna UEs to transmit or receive multiple parallel data streams by utilizing the large bandwidth available in the unlicensed band. Besides, for all considered systems,  the average
system throughput of all the systems increases monotonically with the maximum DL transmit power $P_{\mathrm{T}}$. Yet, for the systems aiming at maximizing energy efficiency, the corresponding system throughput saturates in the high transmit power
allowance regime, i.e.,  $P_{\mathrm{T}}\geq37$ dBm, in the considered system setting.
  In fact, the energy-efficient SUDAS does not further increase the DL transmit power  since the system throughput gain due to a higher transmit power cannot compensate for
the increased transmit power, i.e., the energy efficiency would decrease. As for the benchmark and baseline systems aiming at  system throughput maximization, the average system throughput increases with the DL transmit power without saturation. For system throughput maximization, the BS always utilizes the entire available DL power budget. Yet, the increased system throughput comes at the expense of a severely degraded system energy efficiency, cf. Figure \ref{fig:EE_PT}.

\vspace*{-4mm}
\subsection{Average System Performance versus Number of SUDACs }\vspace*{-2mm}
Figures \ref{fig:EE_SUDAC} and \ref{fig:CAP_SUDAC} illustrate the average energy efficiency and
throughput versus the number of SUDACs for $N=8$ BS antennas, $P_{\mathrm{T}}=37$ dBm, and different systems. It can be observed that both the system energy efficiency and the system throughput of the proposed SUDAS grows with the number of SUDACs, despite the increased power consumption associated with each additional SUDAC.  For $N\ge M$,   for DL transmission, additional SUDACs  facilitate a more efficient conversion of the spatial multiplexing gain in the  licensed band to a frequency multiplexing gain in the unlicensed band which leads to a   significant data rate improvement.  Similarly, for UL transmission, the SUDACs help in converting the frequency multiplexing gain in the unlicensed band to a spatial multiplexing gain in the licensed band. For $M>N$,  increasing the number of SUDACs in the system
leads to more spatial diversity which also improves energy efficiency and system
throughput. Besides, a substantial performance gain can be achieved by the SUDAS  compared to the baseline system for an increasing number of available SUDACs.

\begin{figure}[t]\vspace*{-6mm}
 \centering
 \begin{minipage}[b]{0.45\linewidth} \hspace*{-1cm}
\includegraphics[width=3.5 in]{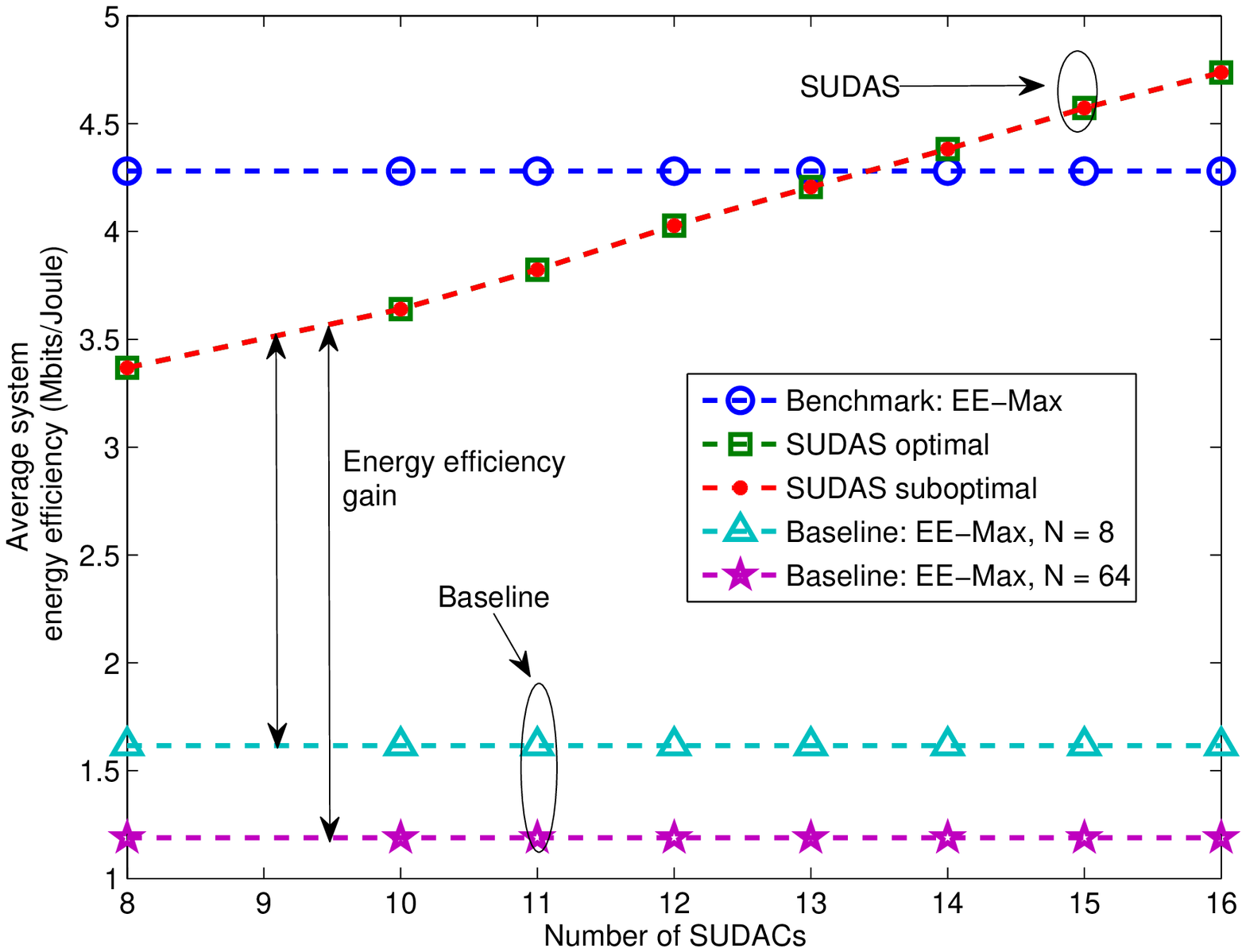}\vspace*{-8mm}
\caption{Average energy efficiency (Mbits/Joule) versus the number of SUDACs.  }\label{fig:EE_SUDAC}
 \end{minipage}\hspace*{1.1cm}
 \begin{minipage}[b]{0.45\linewidth} \hspace*{-1cm}
\includegraphics[width=3.5 in]{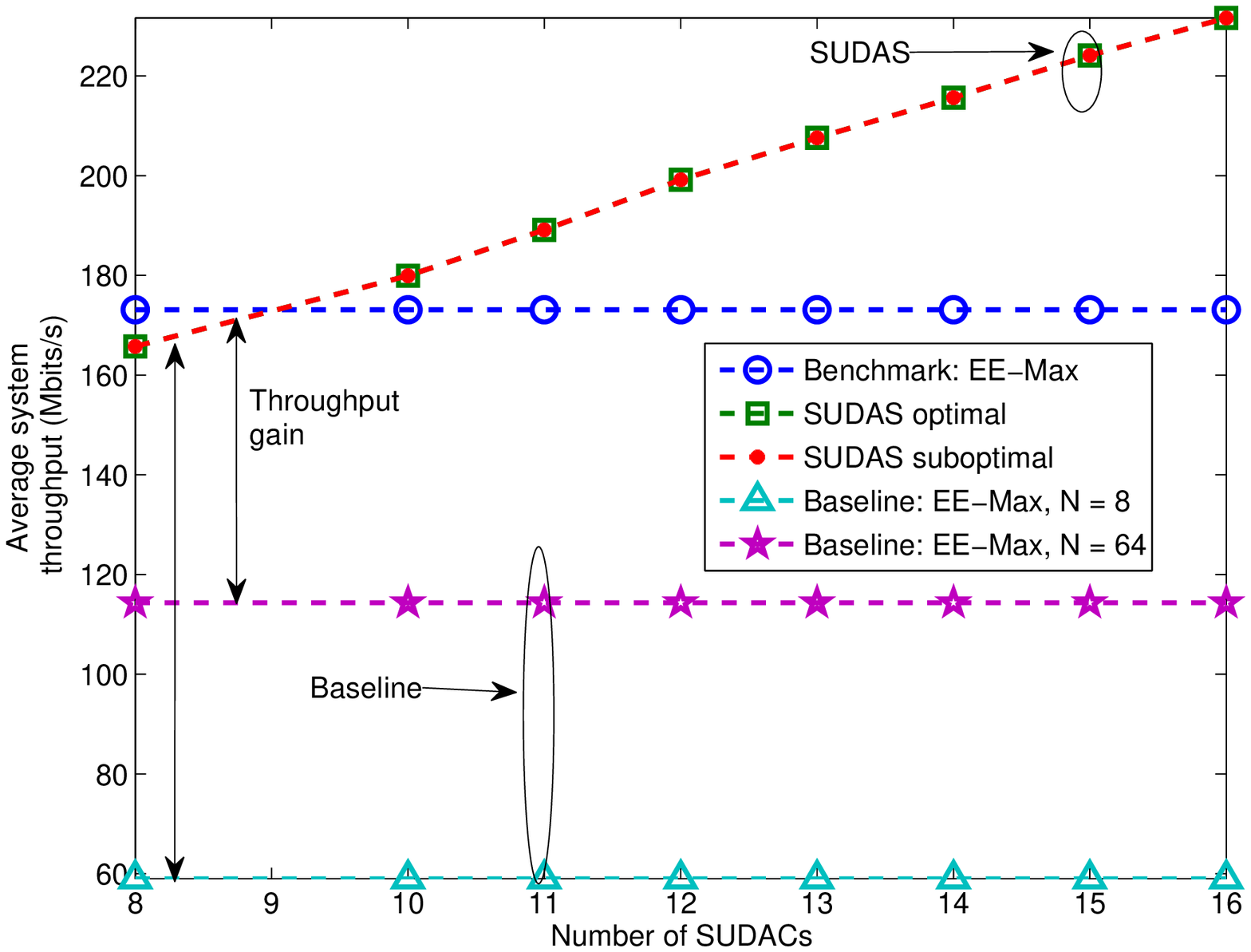}\vspace*{-8mm}
\caption{Average system throughput (Mbits/s) versus  the number of SUDACs. } \label{fig:CAP_SUDAC}
 \end{minipage}\vspace*{-8mm}
\end{figure}
\vspace*{-5mm}
\section{Conclusions}\label{sect:conclusion}
In this paper, we studied the resource allocation algorithm design for SUDAS assisted outdoor-to-indoor communication.
Specifically, the proposed SUDAS  simultaneously utilizes  licensed and unlicensed frequency bands   to facilitate spatial and frequency multiplexing gains for single-antenna UEs in DL and UL, respectively.
 The resource allocation algorithm design was formulated as a non-convex matrix optimization problem. In order to obtain a tractable solution,
we  revealed the structure of the optimal precoding  matrices such that the  problem could be
transformed into a scalar optimization problem. Based on this result, we proposed an asymptotically globally optimal and a suboptimal
iterative  resource allocation algorithm to solve the problem by
alternating optimization. Our simulation results showed that the proposed
SUDAS assisted transmission provides substantial energy efficiency and throughput gains compared to baseline systems which utilize only  the licensed  frequency spectrum for communication.\vspace*{-5mm}
\section*{Appendix-Proof of Theorem \ref{Thm:Diagonalization_optimal}}\vspace*{-2mm}
Due to the page limitation, we provide only a sketch of the proof  which follows a similar approach as in \cite{JR:Yue_Rong_diagonalization,JR:Kwan_FD} and uses majorization theory. We show that the optimal
precoding and post-processing matrices jointly diagonalize the DL and UL end-to-end channel matrices on
each subcarrier for the maximization of the transformed objective function in subtractive form in \eqref{eqn:inner_loop}.
First, we consider the objective function in subtractive form  for  UE $k$ on a per-subcarrier
basis with respect to the optimization variables. In particular, the per-subcarrier objective function for UE $k$ consists of two parts,\vspace*{-2mm}
\begin{eqnarray}
f_1({\cal P},{\cal
S})&=& -s^{[i,k]}_{\mathrm{DL}}\log_2\Big(\det[\mathbf{E}^{[i,k]}_{\mathrm{DL}}]\Big)-s^{[i,k]}_{\mathrm{UL}}\log_2
\Big(\det[\mathbf{E}^{[i,k]}_{\mathrm{UL}}]\Big)\\
f_2({\cal P},{\cal
S})&=&s^{[i,k]}_{\mathrm{DL}}\varepsilon_{\mathrm{B}}
\Tr\Big(\mathbf{P}^{[i,k]}_{\mathrm{DL}}(\mathbf{P}^{[i,k]}_{\mathrm{DL}})^H
\Big) +s^{[i,k]}_{\mathrm{DL}}\varepsilon_{\mathrm{S}}
\Tr\Big(\mathbf{G}^{[i,k]}_{\mathrm{DL}}\Big)+s^{[i,k]}_{\mathrm{UL}}
\varepsilon_{\mathrm{S}}\Tr\Big(\mathbf{G}^{[i,k]}_{\mathrm{UL}}\Big)\notag\\
&+&\varepsilon_k s^{[i,k]}_{\mathrm{UL}}\Tr\Big(\mathbf{P}^{[i,k]}_{\mathrm{UL}}(\mathbf{P}^{[i,k]}_{\mathrm{UL}})^H\Big),
\end{eqnarray}
such that the maximization of the per subcarrier objective function can be expressed as\vspace*{-3mm}
\begin{eqnarray}
\underset{{\cal P},{\cal
S}}{\mino}-f_1({\cal P},{\cal
S})+\eta_{\mathrm{eff}} f_2({\cal P},{\cal
S}).
\end{eqnarray}\vspace*{-1mm}
Besides,  the  determinant of the  MSE matrix on subcarrier $i$ for UE $k$ can be written as
\begin{eqnarray}
\det\Big(\mathbf{E}^{[i,k]}_{\mathrm{DL}}\Big)&=&\prod_{j=1}^{N_{\mathrm{S}}}\Big[ \mathbf{E}^{[i,k]}_{\mathrm{DL}}\Big]_{j,j},
\end{eqnarray}
where  $[\mathbf{X}]_{a,b}$ extracts the $(a,b)$-th element of matrix $\mathbf{X}$.
  $f_1({\cal P},{\cal
S})$ is a Schur-concave function with respect to the optimal precoding matrices
 \cite{JR:Yue_Rong_diagonalization} for a given subcarrier allocation policy $\cal S$.  Thus,   $-f_1({\cal P},{\cal
S})$ is minimized when the MSE matrix $\mathbf{E}^{[i,k]}_{\mathrm{DL}}$ is a diagonal matrix.  Furthermore,  the trace operator in $f_2({\cal P},{\cal
S})$ for the computation of  the total power consumption is also a Schur-concave function with respect to the optimal precoding  matrices. Thus, the optimal  precoding matrices for the minimization of function $f_2({\cal P},{\cal
S})$ should diagonalize the input matrix of the trace function, cf. \cite[Chapter 9.B.1]{book:majorization} and \cite[Chapter 9.H.1.h]{book:majorization}. Similarly, the power consumption functions on the left hand side of  constraints C1--C4 in \eqref{eqn:inner_loop} are also Schur-concave functions and are minimized if the input matrices of the trace functions are diagonal.
    Besides, the non-negative weighted sum of Schur-concave functions over the subcarrier and UE indices preserves Schur-concavity. In other words,  the optimal precoding matrices should jointly diagonalize the     subtractive form    objective function in \eqref{eqn:inner_loop} and simultaneously diagonalize matrices $\mathbf{P}^{[i,k]}_{\mathrm{UL}}(\mathbf{P}^{[i,k]}_{\mathrm{UL}})^H,\mathbf{P}^{[i,k]}_{\mathrm{DL}}(\mathbf{P}^{[i,k]}_{\mathrm{DL}})^H, \mathbf{G}^{[i,k]}_{\mathrm{UL}}$, and $\mathbf{G}^{[i,k]}_{\mathrm{DL}}$. This observation establishes a necessary condition for the structure of the optimal precoding  matrices. Finally, by performing SVD on the channel matrices and after some mathematical manipulations,  it can be verified that the  matrices in \eqref{eqn:matrix_P} and \eqref{eqn:matrix_F} satisfy the optimality condition.
%

%
%

\end{document}